\documentclass[11pt]{article}

\usepackage{amssymb,amsmath,amsthm,mathtools}
\usepackage{array,geometry}  
\usepackage{setspace}
\usepackage{bm,paralist,dsfont,nicefrac,bbm}
\usepackage{etoolbox}
\usepackage[normalem]{ulem}

\makeatletter
\newcommand{\pushright}[1]{\ifmeasuring@#1\else\omit\hfill$\displaystyle#1$\fi\ignorespaces}
\newcommand{\pushleft}[1]{\ifmeasuring@#1\else\omit$\displaystyle#1$\hfill\fi\ignorespaces}
\makeatother

\usepackage{tikz,float,booktabs,pifont,multirow,subcaption}
\usepackage{makecell}
\usetikzlibrary{fit,calc,math,shapes,shapes.multipart,decorations.text,arrows,decorations.markings,decorations.pathmorphing,shapes.geometric,positioning,decorations.pathreplacing, patterns}

\usepackage[driverfallback=hypertex,pagebackref=true,colorlinks]{hyperref}
\hypersetup{linkcolor=[rgb]{.7,0,0}}
\hypersetup{citecolor=[rgb]{0,.7,0}}
\hypersetup{urlcolor=[rgb]{.7,0,.7}}
\renewcommand*{\backref}[1]{}
\renewcommand*{\backrefalt}[4]{%
    \ifcase #1 (Not cited.)%
    \or        (Cited on page~#2)%
    \else      (Cited on pages~#2)%
    \fi}

\usepackage{xcolor}

\usepackage{thmtools,thm-restate}
\usepackage{enumerate,enumitem}
\usepackage{nicematrix}

\allowdisplaybreaks
\renewcommand{\>}{\succ}

\newcommand{\E}{\mathbb{E}}
\newcommand{\EF}[1]{\ifstrempty{#1}{\textrm{\textup{EF}}}{\textrm{\textup{EF{$#1$}}}}}

\newcommand{\EFX}{\textrm{\textup{EFX}}}
\newcommand{\EFXBoundedCharity}{\textrm{\textup{EFX-with-bounded-charity}}}
\newcommand{\EFXcharity}{\textrm{\textup{EFX-with-charity}}}
\newcommand{\EFXUnenviedCharity}{\textrm{\textup{EFX-with-charity}}}
\newcommand{\eps}{\varepsilon}
\renewcommand{\epsilon}{\varepsilon}

\newcommand{\fPO}{\textrm{\textup{fPO}}}

\newcommand{\MMS}{\textrm{\textup{MMS}}}

\newcommand{\PO}{\textup{PO}}

\newcommand{\Prop}[1]{\ifstrempty{#1}{\textrm{\textup{Prop}}}{\textrm{\textup{Prop{$#1$}}}}}

\newcommand{\SDEF}[1]{\ifstrempty{#1}{\textrm{\textup{sd-EF}}}{\textrm{\textup{sd-EF{$#1$}}}}}

\newcommand{\UnifPerm}{\textrm{\textsc{Uniform Permutation}}}

\newcommand{\TPS}{\textrm{\textup{TPS}}}
\newcommand{\V}{\mathcal{V}}
\newcommand{\weakSDprefers}{\succeq^\textup{SD}}
\newcommand{\X}{\mathbf{X}}
\newcommand{\Y}{\mathbf{Y}}
\newcommand{\Z}{\mathbf{Z}}

\usepackage[capitalize, nameinlink, noabbrev]{cleveref}

\usepackage[linesnumbered,ruled,vlined]{algorithm2e}

\newtheorem{theorem}{Theorem}
\newtheorem*{theorem*}{Theorem}

\newtheorem*{definition*}{Definition}

\newtheorem{lemma}{Lemma}
\newtheorem*{lemma*}{Lemma}

\newtheorem*{claim*}{Claim}

\newtheorem*{fact*}{Fact}

\newtheorem*{observation*}{Observation}

\newtheorem*{conjecture*}{Conjecture}

\newtheorem*{corollary*}{Corollary}

\newtheorem*{remark*}{Remark}

\newtheorem*{proposition*}{Proposition}

\newtheorem{example}{Example}
\newtheorem*{example*}{Example}

\allowdisplaybreaks

\title{Best-of-Both-Worlds Guarantees with Fairer Endings}

\author{\begin{tabular}{m{0.12\linewidth}m{0.12\linewidth}m{0.12\linewidth}m{0.12\linewidth}m{0.12\linewidth}m{0.12\linewidth}}
	\multicolumn{3}{c}{\textbf{Telikepalli Kavitha}} & \multicolumn{3}{c}{\textbf{Surya Panchapakesan}} \\
         \multicolumn{3}{c}{\small{TIFR Mumbai}} & \multicolumn{3}{c}{\small{National University of Singapore}} \\
         \multicolumn{3}{c}{\href{mailto:kavitha.telikepalli@gmail.com}{\small{\texttt{kavitha.telikepalli@gmail.com}}}} & \multicolumn{3}{c}{\href{mailto:csurya.panchapakesan@u/nus.edu}{\small{\texttt{csurya.panchapakesan@u.nus.edu}}}} \\
         &&&\\
         \multicolumn{2}{c}{\textbf{Rohit Vaish}} & \multicolumn{2}{c}{\textbf{Vignesh Viswanathan}} & \multicolumn{2}{c}{\textbf{Jatin Yadav}}\\
         \multicolumn{2}{c}{\small{IIT Delhi}} & \multicolumn{2}{c}{\small{University of Massachusetts Amherst}} & \multicolumn{2}{c}{\small{IIT Delhi}}\\
         \multicolumn{2}{c}{\href{mailto:rvaish@iitd.ac.in}{\small{\texttt{rvaish@iitd.ac.in}}}} & \multicolumn{2}{c}{\href{mailto:vviswanathan@umass.edu}{\small{\texttt{vviswanathan@umass.edu}}}} & \multicolumn{2}{c}{\href{mailto:jatin.yadav@cse.iitd.ac.in}{\small{\texttt{jatin.yadav@cse.iitd.ac.in}}}}\\
 	\end{tabular}
 }

\date{} 
\begin{document}

\maketitle
\thispagestyle{empty}
\addtocounter{page}{-1}
\tableofcontents 

\begin{abstract}

We study best-of-both-worlds guarantees for the fair division of indivisible goods. Prior work has shown that under additive valuations, there always exists a randomized allocation that is ex-ante stochastic-dominance envy-free (\SDEF{}) and ex-post envy-free up to one good (\EF{1}). Our work is motivated by the goal of achieving stronger ex-post fairness guarantees such as envy-freeness up to any good (\EFX{}). 

Towards this goal, we make the following contributions:

\begin{itemize}
    \item Our first set of results considers a subdomain of additive valuations called \emph{lexicographic} preferences. On the negative side, we show that ex-ante \SDEF{} is fundamentally in conflict with ex-post \EFX{}, prompting a relaxation in the ex-ante benchmark. On the positive side, we present a polynomial time algorithm that achieves ex-ante $\max\{1-\frac{2}{n+1},\,\frac{9}{10}\}$-\EF{} and ex-post \EFX{} and Pareto optimality, where $n$ is the number of agents. Our algorithm uses a novel dependent rounding technique that, to the best of our knowledge, is its first application in fair division.  
    
    \item Our second set of results concerns beyond-additive valuations. We present a pseudopolynomial time algorithm for \emph{monotone} valuations that achieves $\nicefrac{1}{2}$-\EF{} and ex-post \EFXcharity. For \emph{subadditive} valuations, we achieve ex-ante $\nicefrac{1}{3}$-\EF{} and ex-post $\nicefrac{1}{2}$-\EFX{}. 
\end{itemize}
\end{abstract}

\pagenumbering{roman}

\setcounter{page}{1}
\pagenumbering{arabic}

\section{Introduction} \label{sec:Introduction}

We study a fair division problem with a set of indivisible goods and a set of agents with combinatorial valuations over the goods. In this setting, well-known solution concepts such as \emph{envy-freeness}~\cite{GS58puzzle,F67resource} can fail to exist even with one desirable good and two agents. A standard relaxation is \emph{envy-freeness up to one good} (\EF{1}), which requires that whenever one agent, say Alice, envies another, say Betty, removing \emph{some} good from Betty's bundle eliminates Alice's envy~\cite{LMM+04approximately,B11combinatorial}. The stronger requirement of \emph{envy-freeness up to any good} (\EFX{}) asks this to hold after removing \emph{any} good from Betty's bundle~\cite{CKM+19unreasonable}. 

The study of approximate fairness notions such as \EF{1} and \EFX{} has led to a rich literature~\cite{AAB+23fair}. However, such allocations might still seem unfair to envious agents. Consider, once again, the two-agent one-good example. If Betty receives the good, then Alice might consider the allocation unfair even though it is \EFX{}. A natural approach to improve fairness is to use \emph{randomization}, i.e., each of the two agents will be chosen with probability
$1/2$ and the good is given to the chosen agent.

Thus with $n$ agents, one can imagine giving the entire set of goods to each individual agent with probability $1/n$. While such a randomized allocation is certainly envy-free in expectation (i.e., it is \emph{ex-ante} fair), it is highly unfair \emph{ex-post}, i.e., in the final allocation. In order to achieve fairness both in expectation and \emph{ex-post},
the best-of-both-worlds framework was introduced by Aziz et al.~\cite{AFS+24best}. Formally, the goal is to compute a probability distribution $\langle p_1,\ldots,p_{q}\rangle$ over a set of allocations $A^1,\ldots,A^{q}$ such that not only is each of $A^1,\ldots,A^{q}$ approximately fair ({\em ex-post} fairness) but the randomized allocation 
$\X = \{(p_1,A^1),\ldots,(p_{q},A^{q})\}$, where allocation $A^\ell$ is picked with probability~$p_\ell$,
is fair in expectation (called {\em ex-ante} fairness).
They showed for {\em additive valuations},\footnote{For additive valuations, the value of any $S \subseteq M$ is the sum of values of goods in $S$.} there is always a randomized allocation that is ex-ante {\em sd-envy-free}\footnote{\SDEF{} refers to stochastic dominance-based envy-freeness and is a stronger fairness notion than envy-freeness (\EF{}); see \Cref{sec:prelims} for a formal definition.} and ex-post \EF{1}. 

Our work is motivated by the pursuit of stronger ex-post guarantees such as \EFX{}. 
While \EFX{} is arguably the best fairness analog of envy-freeness of indivisible items~\cite{CGH19envy},
recent results highlight fundamental barriers towards achieving \EFX{}. Akrami et al.~\cite{akrami2026counterexampleefxnge} showed that an EFX allocation may fail to exist under {\em monotone valuations} even with just three agents\footnote{Any valuation function $v$ that obeys $v(S) \le v(T)$ for $S \subseteq T \subseteq M$ is monotone.}.
Mackenzie and Suzuki~\cite{MS26efxnonexistence} further established the non-existence of $0.89$-\EFX{} for {\em subadditive valuations}\footnote{ A monotone valuation function $v$ is subadditive if $v(S \cup T) \le v(S) + v(T)$ for any subsets $S,T$ of~$M$. 
}. 

The above findings motivate our central question. 
\begin{center}
\emph{Can approximate \EFX{} be achieved alongside non-trivial ex-ante fairness?}
\end{center}

\subsection{Our Results}\label{sec:results}
Our results show that there exist important valuation classes where
approximate \EFX{} is achievable in conjunction with strong ex-ante fairness, despite the aforementioned 
impossibility barriers on the existence of \EFX{} allocations. These results are summarized in \Cref{fig:relations-lexicographic,fig:relations-subadditive}.
Formally, an allocation is $\alpha$-\EFX{} if for any pair of agents Alice and Betty, Alice's value for her own bundle is at least $\alpha$ times Alice's value for Betty's bundle upon the removal of {\em any} good from Betty's bundle.

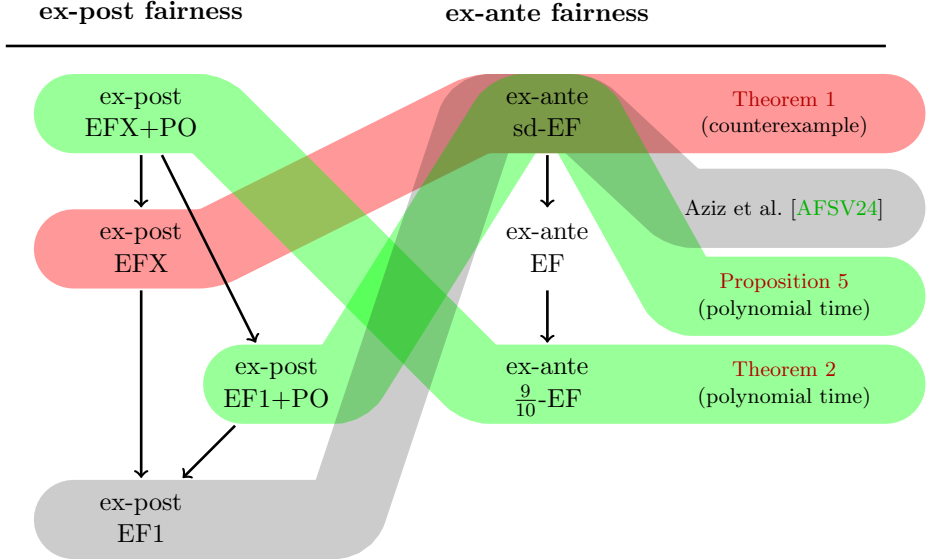
\begin{figure}[t]%
	\begin{center}
	    \scalebox{1}{
	        \begin{tikzpicture}[scale=0.9, every node/.style={scale=0.9}]
	            \tikzstyle{onlytext}=[]
	            \tikzset{venn circle/.style={circle,minimum width=0mm,fill=#1,opacity=0.1}}

	            \draw [line width=30pt,opacity=0.4,color=red,line cap=round,rounded corners] (-4,-4) -- (-2,-4) -- (2,-2) -- (8,-2);
	            \node[onlytext] () at (6.5,-2) {\footnotesize\begin{tabular}{c}{\Cref{thm:SDEF_EFX}}\\{(counterexample)}\end{tabular}};

	            \draw [line width=30pt,opacity=0.4,green,line cap=round,rounded corners] (-1.5,-6) -- (0,-6) -- (2.5,-2) -- (3.5,-2) -- (5,-4.7) -- (8,-4.7);
	            \node[onlytext] () at (6.5,-4.7) {\footnotesize\begin{tabular}{c}{\Cref{prop:sdEF_EF1+PO_lexicographic}}\\{(polynomial time)}\end{tabular}};

	            \draw [line width=30pt,opacity=0.4,green,line cap=round,rounded corners] (-4,-2) -- (-2,-2) -- (2,-6) -- (8,-6);
	            \node[onlytext] () at (6.5,-6) {\footnotesize\begin{tabular}{c}{\Cref{thm:dependent_rounding}}\\{(polynomial time)}\end{tabular}};

                    \draw [line width=30pt,opacity=0.2,black,line cap=round,rounded corners] (-4,-8) -- (0,-8) -- (2,-2) -- (3.5,-2) -- (5,-3.4) -- (8,-3.4);
	            \node[onlytext] () at (6.5,-3.4) {\begin{tabular}{c}{\footnotesize Aziz et al.~\cite{AFS+24best}}\end{tabular}};

                \node[onlytext] (ex-ante fair) at (3,-0.5) {\textbf{ex-ante fairness}};
                \draw[-, line width=1pt] (-5,-1) -- (8,-1);
                \node[onlytext] (ex-ante SD-EF) at (3,-2) {\begin{tabular}{c}{ex-ante}\\{\SDEF{}}\end{tabular}};
                \node[onlytext] (ex-ante EF) at (3,-4) {\begin{tabular}{c}{ex-ante}\\{\EF{}}\end{tabular}};
                \node[onlytext] (ex-ante almostEF) at (3,-6) {\begin{tabular}{c}{ex-ante}\\{$\max\{\frac{9}{10}, \, 1-\frac{2}{n+1}\}$-\EF{}}\end{tabular}};
                \node[onlytext] (ex-post fair) at (-3,-0.5) {\textbf{ex-post fairness}};
                \node[onlytext] (ex-post EFXPO) at (-3,-2) {\begin{tabular}{c}{ex-post}\\{\EFX{}+\PO{}}\end{tabular}};
	            \node[onlytext] (ex-post EFX) at (-3,-4) {\begin{tabular}{c}{ex-post}\\\EFX{}\end{tabular}};
	            \node[onlytext] (ex-post EF1) at (-3,-8) {\begin{tabular}{c}{ex-post}\\{\EF{1}}\end{tabular}};
	            \node[onlytext] (ex-post EF1PO) at (-1,-6) {\begin{tabular}{c}{ex-post}\\{\EF{1}+\PO{}}\end{tabular}};

	            \draw[->, line width=1pt] (ex-ante SD-EF) -- (ex-ante EF);
	            \draw[->, line width=1pt] (ex-ante EF) -- (ex-ante almostEF);
	            \draw[->, line width=1pt] (ex-post EFXPO) -- (ex-post EFX);
	            \draw[->, line width=1pt] (ex-post EFX) -- (ex-post EF1);
	            \draw[->, line width=1pt] (ex-post EFXPO) -- (ex-post EF1PO);
	            \draw[->, line width=1pt] (ex-post EF1PO) -- (ex-post EF1);
	    \end{tikzpicture}
	 }
	\end{center}
	\caption{A summary of our results for \emph{lexicographic} preferences. The arrows denote logical implications between fairness notions. The property combinations known from prior work are shaded in gray, while the positive and negative results shown by us are in green and red, respectively. This figure is inspired by~\cite[Figure~1]{AFS+24best}.}
	\label{fig:relations-lexicographic}
\end{figure}

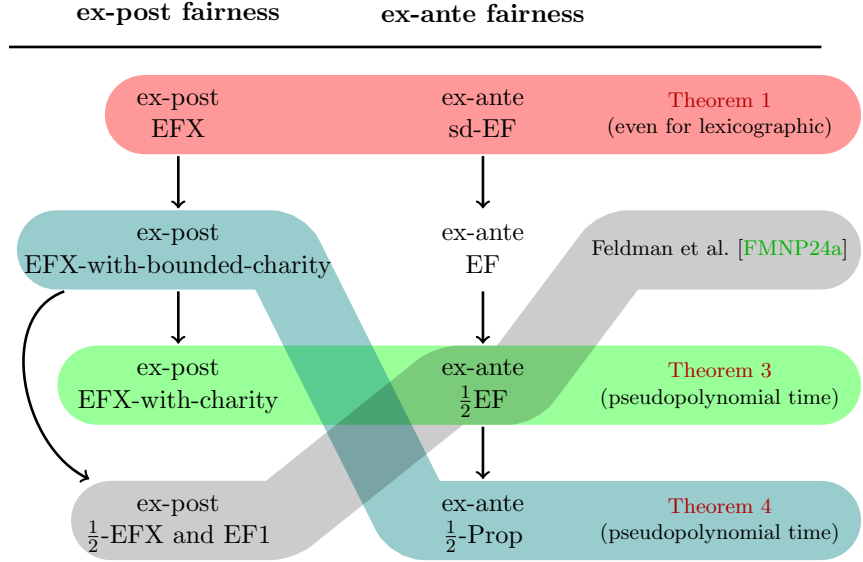
\begin{figure}[t]%
	\begin{center}
	    \scalebox{1}{
	        \begin{tikzpicture}[scale=0.9, every node/.style={scale=0.9}]
	            \tikzstyle{onlytext}=[]
	            \tikzset{venn circle/.style={circle,minimum width=0mm,fill=#1,opacity=0.1}}

	            \draw [line width=30pt,opacity=0.4,color=red,line cap=round,rounded corners] (-2,-2) -- (8,-2);
	            \node[onlytext] () at (6.5,-2) {\footnotesize\begin{tabular}{c}{\Cref{thm:SDEF_EFX}}\\{(even for lexicographic)}\end{tabular}};

	            \draw [line width=30pt,opacity=0.4,green,line cap=round,rounded corners] (-2.7,-6) -- (-0,-6) -- (2,-6) -- (4,-6) -- (5.5,-6) -- (8,-6);
	            \node[onlytext] () at (6.5,-6) %
	            {\footnotesize\begin{tabular}{c}{\Cref{thm:bobw-charity-monotone}}\\{(pseudopolynomial time)}\end{tabular}};

	            \draw [line width=30pt,opacity=0.4,teal,line cap=round,rounded corners] (-3.3,-4) -- (0,-4) -- (2,-8) -- (8,-8);
	            \node[onlytext] () at (6.5,-8) %
	            {\footnotesize\begin{tabular}{c}{\Cref{thm:Ex-ante_half-Prop_ex-post_EFX-with-charity}}\\{(pseudopolynomial time)}\end{tabular}};

	            \draw [line width=30pt,opacity=0.2,black,line cap=round,rounded corners] (-2.5,-8) -- (0,-8) -- (2.5,-6) -- (3.5,-6) -- (5,-4) -- (8,-4);
                \node[onlytext] () at (6.5,-4) {\footnotesize Feldman et al.~\cite{FMN+24breaking}};

                \node[onlytext] (ex-ante fair) at (3,-0.5) {\textbf{ex-ante fairness}};
                \draw[-, line width=1pt] (-4,-1) -- (8,-1);
                \node[onlytext] (ex-ante SD-EF) at (3,-2) {\begin{tabular}{c}{ex-ante}\\{\SDEF{}}\end{tabular}};
                \node[onlytext] (ex-ante EF) at (3,-4) {\begin{tabular}{c}{ex-ante}\\{\EF{}}\end{tabular}};
                \node[onlytext] (ex-ante halfEF) at (3,-6) {\begin{tabular}{c}{ex-ante}\\{$\frac12\EF{}$}\end{tabular}};
                \node[onlytext] (ex-ante half-Prop) at (3,-8) {\begin{tabular}{c}{ex-ante}\\{$\frac{1}{2}$-\Prop{}}\end{tabular}};
                \node[onlytext] (ex-post fair) at (-1.5,-0.5) {\textbf{ex-post fairness}};
                \node[onlytext] (ex-post EFX) at (-1.5,-2) {\begin{tabular}{c}{ex-post}\\{\EFX{}}\end{tabular}};
	            \node[onlytext] (ex-post EFX-with-charity) at (-1.5,-4) {\begin{tabular}{c}{ex-post}\\\EFXBoundedCharity{}\end{tabular}};
	            \node[onlytext] (ex-post half-efx) at (-1.5,-8) {\begin{tabular}{c}{ex-post}\\{$\frac{1}{2}$-\EFX{} and $\EF{1}$}\end{tabular}};
	            \node[onlytext] (ex-post half-fair) at (-1.5,-6) {\begin{tabular}{c}{ex-post}\\{EFX-with-charity}\end{tabular}};
	            \draw[->, line width=1pt] (ex-ante SD-EF) -- (ex-ante EF);

	            \draw[->, line width=1pt] (ex-ante EF) -- (ex-ante halfEF);
	            \draw[->, line width=1pt] (ex-ante halfEF) -- (ex-ante half-Prop);

	            \draw[->, line width=1pt] (ex-post EFX) -- (ex-post EFX-with-charity);
                \draw[->, line width=1pt] (ex-post EFX-with-charity) to [out=200,in=155](ex-post half-efx);
	            \draw[->, line width=1pt] (ex-post EFX-with-charity) -- (ex-post half-fair);
	    \end{tikzpicture}
	 }
	\end{center}
	\caption{A summary of our main results for \emph{subadditive} valuations.  
    The arrows denote logical implications between fairness notions. The property combinations known from prior work are shaded in gray, while the positive and negative results shown by us are in green/teal and red, respectively. This figure is inspired by~\cite[Figure~1]{AFS+24best}.}
	\label{fig:relations-subadditive}
\end{figure}

\paragraph{Lexicographic preferences.} Lexicographic preferences date back to the   work of Hausner~\cite{Hausner54} and have since been widely used to model human decision-making across the social sciences. Qualitatively, lexicographic preferences naturally arise when agents prioritize certain `features' as non-substitutable. Some simple examples include:

\begin{itemize}
    \item A consumer purchasing laptops may first restrict attention to laptops to specific brands. Among these, they may require a minimum of 16GB RAM, and only afterward compare secondary attributes such as price or display quality.

    \item Recruiters may first screen applicants based on mandatory qualifications, and subsequently compare shortlisted candidates using more nuanced criteria, in a hierarchical fashion ~\cite{CS99_job}.
\end{itemize}

Lexicographic preferences also underlie influential theories of behavioral psychology.

\begin{itemize}
    \item Maslow’s famous paper on the \emph{hierarchy of needs}~\cite{Maslow43} organizes human motivations into a hierarchical structure.  Although later works have refined or challenged parts of the theory, it remains a foundational framework often used in sociology research and the management sciences.
    \item Another seminal contribution, Tversky’s model of \emph{elimination by aspects}~\cite{Tversky72}, posits that individuals often simplify complex choices by sequentially eliminating alternatives that lack a desirable `attribute' in a priority-based manner.
\end{itemize}

 More recently, the fair division community has shown interest in lexicographic preferences, in both the indivisible and the divisible settings~\cite{SS14,SV15_Lex}.   
 In fact, Hosseini et al.~\cite{HSV+21fair} show that \EFX{} allocations always exist and can be efficiently computed when agents have lexicographic valuations. Recent work in preference aggregation has also found empirical evidence that voters exhibit lexicographic-like preferences over political candidates. ~\cite{HA24_LexVote}. \\

In our work, we deal with the indivisible setting, with agent preferences over bundles following the lexicographic extensions of their linear ranking over individual goods. Informally, an agent values the goods according to a strict order, say $g_1 \> g_2 \> \dots \> g_m$, such that it prefers any bundle containing $g_1$ over any bundle that doesn't, subject to which, it prefers any bundle containing $g_2$ over any bundle that doesn't, and so on. The reader may refer to \Cref{sec:prelims} for a more formal definition. 

An `optimal' best-of-both-worlds result would be to combine ex-post \EFX{} with the ex-ante guarantee of Aziz et al.~\cite{AFS+24best}, namely {\em sd-envy-freeness} (\SDEF{}). %
We show that such a randomized allocation may not exist
even for lexicographic valuations (\Cref{apdx:impossible}).

\begin{restatable}[\textbf{Incompatibility of ex-ante \SDEF{} and ex-post \EFX{}}]{theorem}{SDEFandEFX}
There exists an instance with lexicographic preferences where no randomized allocation is simultaneously ex-ante $\SDEF{}$ and ex-post $\EFX{}$.
\label{thm:SDEF_EFX}
\end{restatable}

The conflict between ex-ante \SDEF{} and ex-post \EFX{} prompts a relaxation in the ex-ante benchmark, raising the obvious question: {\em Is it possible to achieve ex-ante \EF{} with ex-post \EFX{}?}
Though we do not yet know the answer to this question, we show an equally
positive best-of-both-worlds result by slightly weakening the ex-ante guarantee from \EF{} to 
$\max\{\frac9{10}, 1-\frac{2}{n+1}\}$-\EF{} (where $n$ is the number of agents) and compensating it with the stronger
ex-post guarantee of \EFX{}$+$\PO{}, where \PO{} is Pareto optimality.

\begin{restatable}[\textbf{Ex-ante $\max\{\frac{9}{10}, 1-\frac{2}{n+1}\}$-\EF{} and ex-post \EFX{}+\PO{}}]{theorem}{dependentrounding}
There exists a polynomial time algorithm that, given as input any goods instance with lexicographic valuations, returns a randomized allocation that is ex-ante $\max\{\frac{9}{10}, \,1-\frac{2}{n+1}\}$-\EF{} and ex-post \EFX{} and \PO{}, where $n$ is the number of agents in the instance.
\label{thm:dependent_rounding}
\end{restatable}

\noindent{Thus} for lexicographic valuations, we show it is possible to combine ex-ante {\em almost} envy-freeness, i.e., $\max\{\frac{9}{10}, \frac{n-1}{n+1} \}$-\EF{}, with ex-post almost envy-freeness (EFX) for indivisible goods and economic efficiency, i.e., \EFX{}+\PO{}, establishing the first best-of-both-worlds guarantee in this setting.  

\paragraph{\bf Monotone valuations.} 
 For identical agents, existence of an ex-ante \EF{} and ex-post \EFX{} allocation is immediate: a uniform distribution over the bundles of a Leximin++ allocation suffices. 
 For general monotone valuations, Feldman et al.~\cite{FMN+24breaking} show that any constant approximation of ex-post \EFX{} is incompatible with every constant approximation of ex-ante-\EF{}.  In light of this impossibility result, we consider a relaxation of \EFX{} called \emph{EFX-with-charity} introduced in \cite{CGH19envy}, where some goods may be left unassigned but no agent envies the set of unassigned goods. The unallocated goods in $M \setminus \cup_{i=1}^nA_i$ are said to be donated to \emph{charity}. 

It is known that \EFX-with-charity{} allocations always exist for monotone valuations and can be computed in pseudopolynomial time~\cite{CKM+21little}.
We show the following best-of-both-worlds result for monotone valuations.

\begin{restatable}[\textbf{Ex-ante $\frac{1}{2}$-\EF{} and ex-post \EFX-with-charity{}}]{theorem}{BoBWCharityMonotone}
There exists a pseudopolynomial time algorithm that, given any instance of the fair division problem with monotone valuations,  outputs an allocation that is ex-post \EFX-with-charity{} and ex-ante \nicefrac{1}{2}-\EF.
\label{thm:bobw-charity-monotone}
 \end{restatable}

Theorem \ref{thm:bobw-charity-monotone} is obtained via a simple randomized variant similar to the charity algorithm of ~\cite{CKM+21little}.    
In \EFX-with-charity{} allocations, there is no bound on the number of goods left unassigned, i.e., given to charity.
There is an even stronger variant of \EFX-with-charity{} where only $n-1$ goods are donated to charity~\cite{CKM+21little}.  %
We will call this notion ``\EFXBoundedCharity{}''.
~\cite{CKM+21little} show that for monotone valuations, an \EFXBoundedCharity{} allocation always exists and can be computed in pseudopolynomial time. A natural question is whether the ex-post guarantee in \Cref{thm:bobw-charity-monotone}
can be strengthened to \EFXBoundedCharity{}. We give a partially positive answer to this question for \textit{subadditive} valuations. 

\begin{restatable}[\textbf{Ex-ante $\frac{1}{2}$-\Prop{} and ex-post \EFXBoundedCharity{}}]{theorem}{PSandCharity}
There exists a pseudopolynomial time algorithm that, given as input any goods instance with subadditive valuations, returns a randomized allocation that is ex-ante $\frac{1}{2}\text{-}\Prop{}$ and ex-post \EFXBoundedCharity{}.
\label{thm:Ex-ante_half-Prop_ex-post_EFX-with-charity}
\end{restatable}

Using a careful modification of this result, we show that it is possible to achieve a best of both-worlds result with approximate ex-post EFX.

\begin{restatable}[\textbf{Ex-ante $\frac{1}{3}$-\EF{} and ex-post $\frac{1}{2}$-\EFX{}}]{theorem}{randomsource}
There exists a pseudopolynomial time algorithm that, given as input any goods instance with subadditive valuations, returns a randomized allocation that is ex-ante $\frac{1}{3}\text{-}$\EF{} and ex-post $\frac{1}{2}$-\EFX{}.
\label{thm:random_source_algo}
\end{restatable}

This result nearly matches the state of the art guarantee of Feldman et al.~\cite{FMN+24breaking}, who achieve ex-post $\nicefrac{1}{2}$-\EFX{} in conjunction with ex-ante $\nicefrac{1}{2}$-\EF{}. Our algorithm, however, satisfies a much stronger ex-post property: for any allocation in the support, there exists an agent $s$, such that the allocation satisfies the \EFX{} condition for any pair of agents that does not include $s$.

\subsection{Overview of Our Techniques}

Aziz et al.~\cite{AFS+24best} showed that a randomized allocation that is ex-ante \SDEF{} and ex-post \EF{1} can be achieved by a variant of the Probabilistic Serial (PS) procedure~\cite{BM01new}, also called the ``eating'' algorithm. In this algorithm, all agents simultaneously consume---or, colloquially, `eat'---their favorite available good at a uniform speed. Once a good is fully consumed by a set of agents, each one of those agents proceeds to consume its next favorite available good at the same speed. Note that running the eating algorithm until all goods are consumed results in an \SDEF{} fractional allocation. The eating algorithm is the foundation of many of our results.

\paragraph{Results for lexicographic preferences.}
For lexicographic preferences, Hosseini et al.~\cite{HSV+21fair} have characterized \EFX{} allocations as those where any envied agent receives a single good. Thus, to obtain ex-post \EFX{}, we should avoid assigning multiple goods to a potentially envied agent in the ex-post decomposition.

The eating algorithm appears to be a natural starting point for achieving ex-ante \EF{}. A particularly useful feature of the eating algorithm is its \emph{anytime fairness} property, which means that the fractional partial allocation at any stage of the algorithm is ex-ante envy-free. If we run the eating algorithm to completion, we might be forced to assign multiple goods to envied agents in the ex-post allocations. Thus, we run the eating algorithm for only \emph{one} unit of time, which means that each agent consumes a total fractional mass of one unit.\footnote{It is to be noted that the algorithm of Feldman et al.~\cite{FMN+24breaking} also runs the eating algorithm for one unit of time.}

Each agent $i$ can consume several goods fractionally over the course of the one unit-time of the eating algorithm: let $g_i$ be the \emph{last} (or most recent) good consumed by $i$ during the eating algorithm. Define $L = \bigcup_{i \in N} g_i$. We
form a {\em super-good} `$s$' by unifying all the fractional consumed pieces of goods in $L$ (see \Cref{fig0}). 
Let $k$ (referred to as the {\em capacity}) be the sum of fractions of goods in $L$ that have been consumed by all agents. Note that $k$ is an integer.%

{
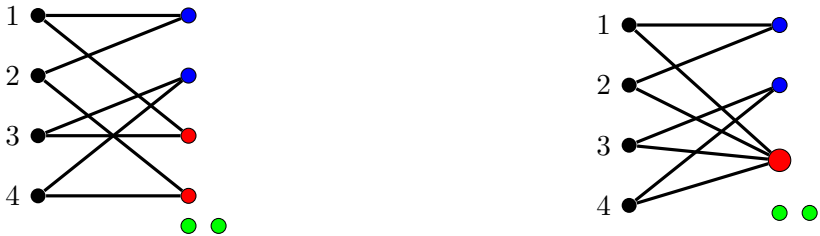
\begin{figure}[h]\centering
\begin{minipage}{0.45\textwidth}
\centering
    \begin{tikzpicture}
    \node[circle,fill=black,label=left:$1$,inner sep = 2pt] (a1) at (0,1.2){};
    \node[circle,fill=black,label=left:$2$,inner sep = 2pt] (a2) at (0,0.4){};
    \node[circle,fill=black,label=left:$3$,inner sep = 2pt] (a3) at (0,-0.4){};
    \node[circle,fill=black,label=left:$4$,inner sep = 2pt] (a4) at (0,-1.2){};

    \node[circle,fill=blue,draw=black,inner sep = 2pt] (b1) at (2,1.2){};
    \node[circle,fill=blue,draw=black,inner sep = 2pt] (b2) at (2,0.4){};
    \node[circle,fill=red,draw=black,inner sep = 2pt] (b3) at (2,-0.4){};
    \node[circle,fill=red,draw=black,inner sep = 2pt] (b4) at (2,-1.2){};
    \node[circle,fill=green,draw=black,inner sep = 2pt] (b5) at (2,-1.6){};
    \node[circle,fill=green,draw=black,inner sep = 2pt] (b6) at (2.4,-1.6){};
     
    \draw[very thick] (a1) -- (b1);
    \draw[very thick] (a1) -- (b3);
    \draw[very thick] (a2) -- (b1);
    \draw[very thick] (a2) -- (b4);
    \draw[very thick] (a3) -- (b3);
    \draw[very thick] (a3) -- (b2);
     \draw[very thick] (a4) -- (b4);
    \draw[very thick] (a4) -- (b2);
    \end{tikzpicture}
\end{minipage}\hspace{0.3cm}
\begin{minipage}{0.45\textwidth}
\centering
    \begin{tikzpicture}
    \node[circle,fill=black,label=left:$1$,inner sep = 2pt] (a1) at (0,1.2){};
    \node[circle,fill=black,label=left:$2$,inner sep = 2pt] (a2) at (0,0.4){};
    \node[circle,fill=black,label=left:$3$,inner sep = 2pt] (a3) at (0,-0.4){};
    \node[circle,fill=black,label=left:$4$,inner sep = 2pt] (a4) at (0,-1.2){};

    \node[circle,fill=blue,draw=black,inner sep = 2pt] (b1) at (2,1.2){};
    \node[circle,fill=blue,draw=black,inner sep = 2pt] (b2) at (2,0.4){};
    \node[circle,fill=red,draw=black,inner sep = 3pt] (b4) at (2,-0.6){};
    \node[circle,fill=green,draw=black,inner sep = 2pt] (b5) at (2,-1.3){};
    \node[circle,fill=green,draw=black,inner sep = 2pt] (b6) at (2.4,-1.3){};
      
    \draw[very thick] (a1) -- (b1);
    \draw[very thick] (a1) -- (b4);
    \draw[very thick] (a2) -- (b1);
    \draw[very thick] (a2) -- (b4);
    \draw[very thick] (a3) -- (b4);
    \draw[very thick] (a3) -- (b2);
     \draw[very thick] (a4) -- (b4);
    \draw[very thick] (a4) -- (b2);
    \end{tikzpicture}
\end{minipage}\hspace{0.3cm}
\caption{The blue goods are not in $L$ while the red ones (they are the last consumed goods) are in $L$ and the green goods are untouched. The figure on the right has a super-good formed by unifying the two red goods.}
\label{fig0}
\end{figure}
}

Let $G$ be the complete bipartite graph on vertex set $N \cup S$, where $N$ is the set of agents and $S$ is the set of all
goods that have been fully consumed during the one unit-time of the eating algorithm, along with the super-good $s$ with 
{\em capacity} $k$. Consider the fractional matching $X = (x_{ij})_{i \in N, j \in S}$ in $G$ that corresponds to the output of the eating algorithm when run for one unit of time. 

\emph{Naive Rounding.} Decomposing $X$ as a convex combination of integral matchings in $G$
(using the Birkhoff-vonNeumann decomposition) and carefully allocating the remaining goods gives us an ex-post $\EFX{}+\PO{}$ guarantee. A careful case analysis of the envy relations yields an ex-ante guarantee of $\frac{6}{7}$-\EF{} (see Algorithm~\ref{alg:utse} and \Cref{thm:utse-upper-bounnd}). More precisely, it gives an ex-ante guarantee of $\frac{3k}{3k+1}$, where $k$ is defined as above (by \Cref{thm:utse-upper-bounnd}). Moreover, when $k = 1$, we can show an exact ex-ante \EF{} guarantee (by \Cref{thm:low-k}). Since $\frac{3k}{3k+1} \ge \frac{9}{10}$ for $k \ge 3$, we have the desired ex-ante guarantee for all $k$, except $k = 2$.

\emph{Dependent Rounding.} For the $k = 2$ case, we employ an alternate method to round the fractional matrix $X$. We use the dependent rounding technique of Gandhi et al.~\cite{GKSADependent} to do this. Dependent rounding is extensively used in approximation algorithms and while it has also appeared in settings such as apportionment and committee selection~\cite{CGST+24,JMW20}; to the best of our knowledge, this constitutes its first application in fair division. We show that in every integral allocation exactly $k$ agents receive the super good $s$. Furthermore, dependent rounding ensures that the probabilities of agents receiving $s$ are negatively correlated. Using this, we are able to improve the ex-ante guarantee to $\frac{9}{10}$-\EF{} this way (see Algorithm \ref{alg:dependent-rounding} and \Cref{thm:k-two} for a detailed explanation). 

In \Cref{sec:1-o(1)}, we further extend this framework by running the eating algorithm for strictly less than one unit of time, specifically-for $1-\frac{1}{n}$ units of time. We show that the resultant fractional matrix can be expressed as convex combination of integral allocations where exactly one agent receives no goods. We then allocate the remaining unallocated items to this agent. This surprisingly simple algorithm gives us an ex-ante $(1-\frac{2}{n+1})$-\EF{} guarantee, where $n$ is the number of agents in the instance. For sufficiently large instances, this improves upon the $\frac{9}{10}$ bound and converges to exact $\EF{}$ (see Algorithm \ref{alg:1-o(1)}).

To summarize, the skeleton of our algorithm is  roughly as follows:
\begin{itemize}
    
    \item Run the eating algorithm for either
    (i) one unit of time, or
    (ii) $1-\frac{1}{n}$ units of time.
    
    \item Sample an integral allocation from the resulting fractional allocation. In the former case, this step can be done via the Birkhoff-von Neumann decomposition (which results in $\frac{6}{7}$-\EF{}) or via dependent rounding when $k=2$ (which achieves $\frac{9}{10}$-\EF{}).
    
    \item Allocate all unallocated goods among agents chosen carefully from the set of unenvied agents in case~(i) and to the lone agent who receives no goods in case~(ii).

\end{itemize}

All the above algorithms are implementable in polynomial time. In \Cref{apdx:1-o(1)}, we exhibit a family of instances for which Algorithm \ref{alg:1-o(1)} cannot do better than ex-ante $\frac{4}{5}$-\EF{}, thereby necessitating taking the maximum over the two approaches.

\paragraph{Results for subadditive and monotone valuations.} 
Our results for subadditive and monotone valuations are based on randomizing the \EFXcharity{} algorithm of Chaudhury et al.~\cite{CKM+21little}. Their algorithm starts with all agents having empty bundles, and repeats the following until the set of unallocated items as a whole is unenvied: choose an inclusion-wise minimal envied subset $Z$ of the unallocated items, an arbitrary agent that envies $Z$ and swap its bundle with $Z$. We show that simply choosing a random agent who envies $Z$ instead of an arbitrary agent who envies $Z$ yields a randomized allocation that is ex-post-\EFXcharity{} and ex-ante $\nicefrac{1}{2}$-\EF{} for monotone valuations (\Cref{thm:bobw-charity-monotone}). 

The standard approach of constructing a fractional allocation and then rounding it; see, for example, the works of Feldman et al.~\cite{FMN+24breaking}, Aziz et al.~\cite{AFS+24best}, and Babaioff et al.~\cite{BEF21best} has severe limitations as we move beyond subadditive to general monotone valuations since rounding algorithms usually cannot capture the synergies between goods in a bundle. This makes proving ex-ante guarantees highly challenging. Sidestepping these issues, we provide a simple algorithm which, to the best of our knowledge is the first best-of-both-worlds guarantee for general monotone valuations.

Additionally, we show that a slightly different randomization yields ex-post $\nicefrac{1}{2}$-\EFX{} and ex-ante $\nicefrac{1}{3}$-\EF{} for subadditive valuations (by \Cref{thm:random_source_algo}). To achieve this, we choose an agent $s$ uniformly at random to serve as the least preferred agent to be chosen to swap its bundle with the minimal-envied-subset. Instead of choosing a random agent who envies the minimal-envied-subset, we choose a random agent from the set of agents excluding $s$ (unless $s$ is the only agent who envies the minimal envied subset). By construction, agent $s$ is never envied throughout the execution. Consequently, once no agent envies the pool of unallocated items, all remaining goods can be assigned to $s$. With a careful analysis, this algorithm can be shown to output random allocations which are ex-post $\nicefrac{1}{2}$-\EFX{} and ex-ante $\nicefrac{1}{3}$-\EF{}.

\paragraph{\bf Organization of the paper.}
\Cref{sec:impossible} contains the impossibility result (\Cref{thm:SDEF_EFX}).
The algorithms concerning \Cref{thm:dependent_rounding} are described in \Cref{sec:Results-Lexicographic}, while the proofs are deferred to the Appendix. 
\Cref{thm:bobw-charity-monotone,thm:random_source_algo} are proved in \Cref{sec:Results-Additive}. We discuss preliminaries in \Cref{sec:prelims}.

\subsection{Background and Related Results}
The twin goals of ex-ante and ex-post fairness in resource allocation problems have been studied traditionally in the economics literature and, more recently, also within computer science. Early works in this direction were in context of the \emph{random assignment} problem, where $n$ objects must be allocated to $n$ individuals with each individual receiving one object. These works include \cite{HZ79efficient}, \cite{BM01new}, and \cite{BCK+13designing}.

Closest to our work are those of Aziz et al.~\cite{AFS+24best}, Babaioff et al.~\cite{BEF21best}, Feldman et al.~\cite{FMN+24breaking} and Bu et al~\cite{BLL+24mixed}. 
\begin{itemize}
    \item The question of whether an ex-ante sd-EF and ex-post \EF{1} allocation always exists was formulated by Freeman et al.~\cite{FSV20best}, who showed a positive answer by means of a variant of the Probabilistic Serial procedure. Shortly after, Aziz~\cite{A20simultaneously} showed that this result can, in fact, be achieved through the Probabilistic Serial procedure itself. Note that \cite{AFS+24best} is the merger of the works of Freeman et al.~\cite{FSV20best} and Aziz~\cite{A20simultaneously}.
    \item The follow-up work of Babaioff et al.~\cite{BEF21best} took a complementary approach by asking for \emph{share-based} (instead of envy-based) ex-post fairness. They provided an algorithm for computing an ex-ante \Prop{} and ex-post half truncated proportional share, i.e., $\frac{1}{2}$-\TPS{} (and thus, $\frac{1}{2}$-\MMS{}) allocation.\footnote{\TPS{} is a fair share measure defined in \cite{BEF21best} that is stronger than \MMS{}.}
    \item Recently, it was shown by Feldman et al.~\cite{FMN+24breaking} that there is a 
    randomized allocation that is ex-ante $\frac{1}{2}$-\EF{} and ex-post \EF{1} and $\frac{1}{2}$-\EFX{};\footnote{An allocation $(A_1,\ldots,A_n)$ is $\frac{1}{2}$-\EFX{} if for any pair of agents $i$ and $j$: $v_i(A_i) \ge \frac{1}{2}v_i(A_j\setminus\{g\})$ for any $g \in A_j$.} furthermore, such an allocation can be computed in polynomial time.
    
    \item More recently, Bu et al~\cite{BLL+24mixed} showed that for agents with bi-valued\footnote{for each $i \in N$ and $g \in M$, $v_i(g)\in \{a,\,b\}$ where $0 \leq a < b$  } additive utilities, there exists a randomized allocation that is ex-ante \EF{} and ex-post \EFX{} and \fPO{},\footnote{Fractional Pareto optimality (fPO) is a stronger notion of economic efficiency than Pareto optimality (PO). An integral allocation $(A_1,\ldots,A_n)$ is fPO if there is no fractional allocation $X= (X_1,\ldots,X_n)$ such that $v_i(X_i) \geq v_i(A_i)$ for all $i \in [n]$ and $v_j(X_j) > v_j(A_j)$ for some $j \in [n]$.} and can in fact be computed in polynomial time.
    
\end{itemize}

Best-of-both-worlds guarantees when agents have priorities assigned to them, also called {\em Best-of-both-worlds with entitlements}, were studied in \cite{AGM23entitlements,HSV24entitlements}.
Other related works have considered binary additive valuations~\cite{AAG+15online} and submodular valuations with binary marginals~\cite{BEF21fair}, the design of fractional maximum Nash welfare allocation~\cite{HPP+20fair}, the study of ex-ante and ex-post fairness in an online allocation setting~\cite{MNR21fair}, as well as settings with a mixture of divisible and indivisible goods~\cite{BLL+24mixed}. Additionally, Caragiannis et al.~\cite{CKK21interim} have studied \emph{interim envy-freeness} which is a fairness notion that is stronger than ex-ante but weaker than ex-post envy-freeness.

\section{Preliminaries}
\label{sec:prelims}

\paragraph{Problem instance.} An instance $\langle N, M, \V \rangle$ of the fair division problem is defined by a set of $n$ \emph{agents} $N \coloneqq \{1,2,\dots,n\}$, a set of $m$ \emph{indivisible goods} $M \coloneqq \{g_1,\dots,g_m\}$, and a \emph{valuation profile} $\V \coloneqq \{v_1,v_2,\dots,v_n\}$ that specifies the cardinal preferences of each agent $i \in N$ via a \emph{valuation function} $v_i:2^M \rightarrow \mathbb{R}^{\geq 0}$. We assume $v_i(\emptyset)=0$. For simplicity, we will write $v_i(g)$ in place of $v_i(\{g\})$ to denote agent $i$'s numerical value for good $g$. We will also use $g \>_i g'$ to denote $v_i(g) > v_i(g')$, and $g \succeq_i g'$ to denote $v_i(g) \geq v_i(g')$.

\paragraph{Valuation classes.} 

A valuation function $v_i$ is \emph{monotone} if for any $S \subseteq T \subseteq M$, $v_i(S) \le v_i(T)$. It is said to be \emph{subadditive} if it is monotone and for any $S, T \subseteq M$, $v_i(S \cup T) \le v_i(S) + v_i(T)$. It is 
said to be \emph{lexicographic} if all individual goods have distinct values and for any pair of distinct bundles $S,T \subseteq M$, the following holds: $v(S) > v(T) \text{ if and only if there exists a good } g \in S \setminus T \text{ such that } \{g' \in T : v(g') > v(g)\} \subseteq S.$ Observe that a lexicographic valuation function induces a strict total order over the set of bundles $2^M$.  The hierarchy of these function classes is shown below.
\begin{align*}
\text{lexicographic} \subset \text{additive} \subset \text{subadditive} \subset \text{monotone}
\end{align*}

\noindent\emph{Fractional and randomized allocations.}
A \emph{fractional allocation} of the set $M$ of goods among the agents in $N$ is given by a non-negative $n \times m$ matrix $X$ such that $\sum_{i} X_{i, j} \le 1$ for each $j \in [m]$; $X_{i,j}$ denotes the fraction of good $g_j$ assigned to agent $i$. We say $X$ is \emph{complete} if  $\sum_{i} X_{i, j} = 1$ for all $j \in [m]$ and \emph{partial} otherwise. An \emph{integral} allocation is a special case with the matrix entries in $\{0,1\}$. For an integral allocation $A$, we will write $A_i$ to denote the set of goods (or the \emph{bundle}) assigned to agent $i$.

A \emph{randomized} allocation is a probability distribution over integral allocations, and is specified by a collection of $q$ ordered pairs $\X \coloneqq \{(p_1,A^1),\dots,(p_q,A^q)\}$ such that $\sum_{\ell \in [q]} p_\ell = 1$ and $p_\ell \geq 0$ for all $\ell \in [q]$; here $p_\ell$ is the probability that the realized allocation is $A^\ell$. The set $\{A^1,\dots,A^q\}$ is called the \emph{support} of the randomized allocation $\X$. When agents have additive valuations, one can naturally associate a fractional allocation~$X \coloneqq \sum_{\ell \in [q]} p_\ell A^\ell$ 
to a randomized allocation~$\X$. We will use the term ``allocation'' to refer to an integral allocation; otherwise, we will specify ``fractional'' or ``randomized'' allocation\footnote{Throughout the paper, we use unbolded capital letters such as $A, B, Z$ and $R$ to denote integral allocations while $X$ and $Y$ are reserved for fractional allocations. Boldface capital letters such as $\Z$ are used to denote randomized allocations.}. \\ 
Here, we note that some of our algorithms output randomized allocations (i.e., the entire distribution) while some other algorithms output a `sample' from the distribution. The `sample' is generated via some internal randomness of the algorithm. 

\paragraph{Birkhoff-von Neumann (B-vN) decomposition.}

Some of our algorithms will use the Birkhoff-von Neumann decomposition to decompose fractional allocations into integral allocations. The following is a generalization of the original theorem which can be applied to singly stochastic matrices.

\begin{restatable}[\cite{B46three,vN53certain}]{proposition}{BvNgeneral}
Let $X$ be an $n \times m$ matrix with nonnegative entries such that, for every row $i \in [n]$, $\sum_{j=1}^m X_{i,j} = 1$ and every column $j \in [m]$, $\sum_{i=1}^n X_{i,j} \leq 1$. Then, for some $q \in \mathbb{N}$, there exist $n \times m$ matrices $A^1,\dots,A^q$ and a set of nonnegative numbers $p_1,\dots,p_q$, where $\sum_{\ell \in [q]} p_\ell = 1$ and $X = \sum_{\ell \in [q]} p_\ell A^\ell$ such that for every $\ell \in [q]$,
\begin{itemize}
    \item the elements of $A^\ell$ are in $\{0,1\}$,
    \item for every $i \in [n]$, $\sum_{j=1}^m A^\ell_{i,j} = 1$, and
    \item for every $j \in [m]$, $\sum_{i=1}^n A^\ell_{i,j} \leq 1$.
\end{itemize}
Moreover, $q$ is $\mathcal{O}((m+n)^2)$ and the matrices $A^1,\dots,A^q$ and the coefficients $p_1,\dots,p_q$ can be computed in polynomial time.
\label{prop:BvN_general}
\end{restatable}

\subsection*{Fairness and Efficiency Notions}
An allocation $A = (A_1,\dots,A_n)$ is said to be \emph{envy-free} (\EF{}) if for every pair of agents $i,j \in N$, $v_i(A_i) \geq v_i(A_j)$~\cite{GS58puzzle,F67resource}. An allocation $A$ is \emph{envy-free up to one good} (\EF{1}) if for every pair of agents $i,j \in N$ such that $A_j \neq \emptyset$, we have $v_i(A_i) \geq v_i(A_j \setminus \{g\})$ for some good $g \in A_j$~\cite{LMM+04approximately,B11combinatorial}. An allocation $A$ is \emph{envy-free up to any good} (\EFX{}) if for every pair of agents $i,j \in N$ and for every good $g \in A_j$, we have $v_i(A_i) \geq v_i(A_j \setminus \{g\})$~\cite{CKM+19unreasonable}. \\ An allocation $A$ is $\alpha$-\EF{} for a given $\alpha \in [0,1]$ if for all $i,j\in N$, we have $v_i(A_i)\geq \alpha \cdot v_i(A_j).$ When $\alpha=1$, we write \EF{} instead of 1-\EF{}. These multiplicative approximations are correspondingly defined for the other notions of fairness.

\paragraph{EFX-with-charity.}
An  $n+1$-partition $(A_1,\ldots,A_n,P)$ of the goods $M$ is \EFXUnenviedCharity{} if
(i) the partial allocation $(A_1,\ldots,A_n)$ is \EFX{}, and
(ii) no agent envies the charity bundle $P$, i.e., $v_i(A_i)\ge v_i(P)$ for all $i\in N$.

\paragraph{EFX-with-bounded-charity.}
An allocation is \EFXBoundedCharity{} if it is \EFXUnenviedCharity{} and $|P|$ is smaller than the number of unenvied agents; in particular, $|P|\le n-1$.\footnote{Without loss of generality, fewer than $n$ agents are envied: otherwise the envy graph contains a cycle, which can be eliminated via cyclic swaps while preserving \EFX{}.} An \EFXBoundedCharity{} partition entails additional fairness guarantees. It can be transformed into a {\em complete allocation} (so $P=\emptyset$) that is $\frac{1}{2}$-\EFX{} and \EF{1} \cite[Theorem~3.5]{CKM+21little}. 

\paragraph{sd-Envy-freeness.}
 Given fractional allocations $X$ and $Y$, we say that agent $i$ (weakly) \emph{sd-prefers} $X_i$ to $Y_i$, denoted as $X_i \weakSDprefers_i Y_i$, if for every good $g \in M$, we have $ \sum_{g_j \, : \, g_j \, \succeq_i \, g} X_{i,j} \geq \sum_{g_j \, : \, g_j \, \succeq_i \, g} Y_{i,j} $. This implies that for \emph{any} cardinal additive valuation function $v_i$ that is consistent with the ordinal preferences $\succeq_i$, we have $v_i(X_i) \geq v_i(Y_i)$. A fractional allocation $X$ is called \emph{sd-envy-free} (\SDEF{}) if for every pair of agents $i,j \in N$, we have $X_i \weakSDprefers_i X_j$. Thus, $\SDEF{} \Rightarrow \EF{}$.

\paragraph{Pareto optimality.} An allocation $A$ is said to Pareto dominate another allocation $B$ if for every agent $i\in N$, we have $v_i(A_i)\geq v_i(B_i)$ and for some agent $j \in N, v_j(A_j) > v_j(B_j)$.  A \emph{Pareto optimal} (\PO{}) allocation is one that is not Pareto dominated by any other allocation.

\paragraph{Ex-ante and Ex-post Fairness.} When agents have additive valuations, a randomized allocation $\X \coloneqq \{(p_1,A^1),\cdots,(p_q,A^q)\}$ is ex-ante $\alpha$-\EF{} (or more generally, it satisfies a property $P$ \emph{ex-ante}) 
if the associated fractional allocation $X \coloneqq \sum_{\ell \in [q]} p_\ell A^\ell$ 
is ex-ante $\alpha$-\EF{} (resp., satisfies $P$). 
For general monotone valuations, ex-ante $\alpha$-\EF{} means
$\mathbb{E}_{A \sim \X}[v_i(A_i)] \ \geq \ \alpha\cdot\mathbb{E}_{A\sim \X}[v_i(A_j)]  \text{ for every } i,j \in N.$

Similarly, we say that $\X$ satisfies a property $Q$ \emph{ex-post} if every integral allocation in the support of $\X$ satisfies $Q$; see \Cref{fig:relations-subadditive}.

\subsection*{Probabilistic Serial (or Eating) Algorithm} The probabilistic serial algorithm~\cite{BM01new} is a well-known method for computing a fractional sd-envy-free (\SDEF{}) allocation when agents have additive valuations,. Agents simultaneously consume their favorite goods at a uniform speed of one good per unit time. When a good is completely consumed by a set of agents, each of those agents then moves on to consume their next preferred available good at the same speed. The algorithm concludes when all goods have been consumed. The result is a fractional allocation in which agents are assigned the fraction of each good that they consumed.

\section{An Impossibility Result}
\label{sec:impossible}
We first establish the incompatibility between ex-ante sd-envy-freeness and ex-post \EFX{} under lexicographic preferences.
Let us first recall what is {\em sd-envy-freeness}.
Given the \emph{ordinal} preferences of agents, \SDEF{} requires that the fractional allocation be envy-free under all cardinal and additive valuations consistent with the ordinal preferences. 
\begin{itemize}
    \item Formally, given a pair of fractional allocations $X$ and $Y$, we say that agent $i$ (weakly) \emph{sd-prefers} $X_i$ to $Y_i$, denoted as $X_i \weakSDprefers_i Y_i$, if for every good $g \in M$, we have $ \sum_{g_j \, : \, g_j \, \succeq_i \, g} X_{i,j} \geq \sum_{g_j \, : \, g_j \, \succeq_i \, g} Y_{i,j} $. 
\end{itemize}
It is known that $X_i \weakSDprefers_i Y_i$ implies that for \emph{any} cardinal additive valuation function $v_i$ consistent with the ordinal preferences $\succeq_i$, we have $v_i(X_i) \geq v_i(Y_i)$.
Our goal is to prove the following result.

\SDEFandEFX*

Towards proving the above result, we will use the following characterization of \EFX{} allocations for lexicographic preferences.

\begin{restatable}[\cite{HSV+21fair}]{proposition}{EFXCharacterization}
When agents have lexicographic preferences, an allocation is \EFX{} if and only if every envied agent receives exactly one good.
\label{prop:EFX_characterization}
\end{restatable}

\begin{proof}[Proof Sketch]
Consider the following instance with four goods $g_1,\dots,g_4$ and three agents $1,2,3$ with lexicographic preferences:
\begin{align*}
1: ~&~ g_1 \> g_3 \> g_4  \> g_2\\
2: ~&~ g_1 \> g_2 \> g_4  \> g_3\\
3: ~&~ g_2 \> g_3 \> g_4  \> g_1
\end{align*}
We use the ordinal representation of the lexicographic preferences in the above instance. Thus, for example, agent 1 prefers any bundle containing good $g_1$ over any bundle that does not, subject to that it prefers any bundle containing $g_3$ over any bundle without it, and so on.

Using \Cref{prop:EFX_characterization}, we deduce that there are four possible EFX allocations in this instance:
\begin{align*}
    \NiceMatrixOptions{code-for-first-row=\scriptstyle,code-for-first-col=\scriptstyle}
    A^1{}=\begin{bNiceMatrix}[first-row,first-col]
      & g_1 & g_2 & g_3 & g_4 \\
      1 & 1 & \cdot & \cdot & \cdot \\
      2 & \cdot & 1 & \cdot & \cdot \\
      3 & \cdot & \cdot & 1 & 1 
  \end{bNiceMatrix}, \qquad 
  A^2{}=\begin{bNiceMatrix}[first-row,first-col]
      & g_1 & g_2 & g_3 & g_4 \\
      1 & 1 & \cdot & \cdot & \cdot \\
      2 & \cdot & \cdot & 1 & 1 \\
      3 & \cdot & 1 & \cdot & \cdot 
  \end{bNiceMatrix}, 
  \\
  \NiceMatrixOptions{code-for-first-row=\scriptstyle,code-for-first-col=\scriptstyle}
  A^3{}=\begin{bNiceMatrix}[first-row,first-col]
      & g_1 & g_2 & g_3 & g_4 \\
      1 & \cdot & \cdot & 1 & \cdot \\
      2 & 1 & \cdot & \cdot & \cdot \\
      3 & \cdot & 1 & \cdot & 1 
  \end{bNiceMatrix}, \qquad
  A^4{}=\begin{bNiceMatrix}[first-row,first-col]
      & g_1 & g_2 & g_3 & g_4 \\
      1 & \cdot & \cdot & 1 & 1 \\
      2 & 1 & \cdot & \cdot & \cdot \\
      3 & \cdot & 1 & \cdot & \cdot 
  \end{bNiceMatrix}.
\end{align*}
No randomized allocation supported by these four EFX allocations is sd-EF. Complete details can be found in \Cref{apdx:impossible}.
\end{proof}

\section{Positive Results for Lexicographic Valuations}
\label{sec:Results-Lexicographic}

In this section we first discuss two natural approaches for simultaneously achieving ex-ante \EF{} and ex-post $\EFX{}+\PO{}$ under lexicographic preferences. The first approach involves uniformly randomizing over round-robin sequences, while the second approach concerns the `eating' algorithm~\cite{AFS+24best}. We show that the former approach fails to be ex-ante envy-free, while the latter fails ex-post \EFX{}. However, as we will show in subsequent sections, combining essential ideas from each of these approaches will take us closer to our desired goal of ex-ante \EF{} and ex-post $\EFX{}+\PO{}$.

\subsection*{Achieving ex-post \EFX{}+\PO{}}
For lexicographic preferences, Hosseini et al.~\cite{HSV+21fair} characterized the class of $\EFX{}+\PO{}$ allocations using \emph{picking sequences}. A picking sequence is a sequence $(a_1, a_2, \dots,a_m)$ of length of $m$, where $a_i \in N$ for every $i \in [m]$. A picking sequence naturally induces an allocation as follows: Agents take turns following the sequence, with each picking its favorite remaining item. For example, given an instance with four goods, the sequence $\sigma = (1, 2, 3, 2)$ corresponds to an allocation where agents $1$, $2$, and $3$ pick their favorite remaining good in that order, and finally agent~$2$ picks the leftover good. It can be observed that under lexicographic preferences, an allocation is \PO{} if and only if it is induced by some picking sequence~\cite{ABL+19efficient,HSV+21fair}.

For any permutation $\sigma$ of the agents in $N$, Hosseini et al.~\cite{HSV+21fair} defined a family of picking sequence-based algorithms as follows: First, construct a partial allocation of $n$ goods by executing the picking sequence $\sigma$. That is, a partial allocation, say $A$, of $n$ goods is constructed by allowing each agent to pick its favorite remaining good according to $\sigma$. Next, to allocate the remaining $(m-n)$ goods, pick an arbitrary picking sequence of length $(m-n)$ of the \emph{unenvied} agents in $A$. We call the resulting $m$-length picking sequence a \emph{$\sigma$-unenvied sequence}.

Hosseini et al.~\cite{HSV+21fair} showed that for any choice of permutation $\sigma$, the resulting $\sigma$-unenvied sequence induces an $\EFX{}+\PO{}$ allocation. Furthermore, every $\EFX{}+\PO{}$ allocation for the given instance is achieved by some picking sequence in the aforementioned family (i.e., a $\sigma$-unenvied sequence for some permutation $\sigma$). \Cref{prop:EFX+PO_characterization} recalls this result.

\begin{restatable}[\cite{HSV+21fair}]{proposition}{EFXPOCharacterization}
For lexicographic preferences and any permutation $\sigma$ of the agents, the allocation computed by any $\sigma$-unenvied sequence satisfies \EFX{} and
\PO{}. Conversely, any $\EFX{}+\PO{}$ allocation can be computed by some $\sigma$-unenvied sequence for some choice of $\sigma$.
\label{prop:EFX+PO_characterization}
\end{restatable}

Thus, in order to achieve ex-post $\EFX{}+\PO{}$ for lexicographic preferences, it suffices to construct a probability distribution over $\sigma$-unenvied sequences. We examine a specific $\sigma$-unenvied sequence where, after the initial picking phase according to $\sigma$, the remaining $(m-n)$ goods are assigned to the last agent in $\sigma$.

\subsubsection*{The Uniform Permutation Algorithm}

Consider an algorithm that takes as input any instance with lexicographic valuations and constructs a randomized allocation as follows: First, a permutation $\sigma$ of agents in $N$ is chosen uniformly at random. The agents take turns according to $\sigma$, and on its turn, each agent picks its favorite remaining item. Finally, all remaining items are assigned to the last agent according to $\sigma$. We call this \UnifPerm{} algorithm.\footnote{When the number of items equals the number of agents, this algorithm is also known as \emph{random priority} or \emph{random serial dictatorship} in the economics literature~\cite{BM01new}.}

It can be seen from \Cref{prop:EFX+PO_characterization} that the above algorithm is ex-post $\EFX{}+\PO{}$. At first glance, one might expect that uniformly randomizing over all permutations would result in ex-ante \EF{}. However, \Cref{eg:Unif_Perm_Fails_3/4-EF} shows that for any $\eps>0$, \UnifPerm{} fails to achieve ex-ante $(\frac{3}{4}+\eps)$-\EF{}. The intuition behind this failure is as follows: For two agents, $i$ and $j$, while the probability of $i$ going before or after $j$ is the same, the accumulated \emph{value deficit} for agent $i$ in the latter case could still outweigh any \emph{advantage} gained in the former.

\begin{example}[\UnifPerm{} fails ex-ante $(\frac{3}{4}+\eps)$-\EF{}]
Consider an instance with six goods and six agents. For $\eps > 0$ sufficiently small, the valuations are given as follows:
\begin{table}[h]
    \centering
    \begin{tabular}{c|cccccc}
        & $g_1$ & $g_2$ & $g_3$ & $g_4$ & $g_5$ & $g_6$\\  
        \cmidrule{1-7}
        $1$ & $1+16\epsilon$ & $1$ & $8\epsilon$ & $4\epsilon$ & $2\epsilon$ & $\epsilon$\\  
        $2$ & $1$ & $1+16\epsilon$ & $8\epsilon$ & $4\epsilon$ & $2\epsilon$ & $\epsilon$\\
        $3$ & $1$ & $\epsilon$ & $8\epsilon$ & $4\epsilon$ & $2\epsilon$ & $1+16\epsilon$\\
        $4$ & $1$ & $\epsilon$ & $8\epsilon$ & $4\epsilon$ & $2\epsilon$ & $1+16\epsilon$\\
        $5$ & $1$ & $\epsilon$ & $8\epsilon$ & $4\epsilon$ & $2\epsilon$ & $1+16\epsilon$\\
        $6$ & $1$ & $\epsilon$ & $8\epsilon$ & $4\epsilon$ & $2\epsilon$ & $1+16\epsilon$
    \end{tabular}
    \label{tab: 0.75 ex-ante EF example}
\end{table}

Let $\X$ denote the randomized allocation computed by the \UnifPerm{} algorithm. The associated fractional allocation $X$ is given by (note that there are $6! = 720$ permutations):
\begin{table}[H]
    \centering
    \begin{tabular}{c|cccccc}
        & $g_1$ & $g_2$ & $g_3$ & $g_4$ & $g_5$ & $g_6$\\  
        \cmidrule{1-7}
        $1$ & $288/720$ & $144/720$ & $72/720$ & $96/720$ & $120/720$ & $0$ \\
        $2$ & $0$ & $576/720$ & $24/720$ & $48/720$ & $72/720$ & $0$ \\
        $3$ & $108/720$ & $0$ & $156/720$ & $144/720$ & $132/720$ & $180/720$\\
        $4$ & $108/720$ & $0$ & $156/720$ & $144/720$ & $132/720$ & $180/720$\\
        $5$ & $108/720$ & $0$ & $156/720$ & $144/720$ & $132/720$ & $180/720$\\
        $6$ & $108/720$ & $0$ & $156/720$ & $144/720$ & $132/720$ & $180/720$
    \end{tabular}
\end{table}

The ratio of agent $1$'s value for its own bundle and that of agent $2$'s bundle under the fractional allocation $X$ is given by:
\begin{align*}
\frac{\E[v_1(X_1)]}{\E[v_1(X_2)]} &= \frac{288 \cdot v_1(g_1) + 144 \cdot v_1(g_2) + 72 \cdot v_1(g_3) + 96 \cdot v_1(g_4) + 120 \cdot v_1(g_5)}{576 \cdot v_1(g_2) + 24 \cdot v_1(g_3) + 48 \cdot v_1(g_4) + 72 \cdot v_1(g_5)} \\ 
&= \frac{432 + 5808\eps}{576+528\eps} = \frac{3}{4} + 5412\eps,
\end{align*}
which is an increasing function in $\eps$ with a minima of $0.75$. Thus, for any given $\eps' > 0$, by choosing $\eps > 0$ to be sufficiently small, we observe that $\X$ fails ex-ante $(\frac{3}{4}+\eps')$-\EF{}.\qed
\label{eg:Unif_Perm_Fails_3/4-EF}
\end{example}

Despite the failure of \UnifPerm{} in achieving better than ex-ante $\frac{3}{4}$-\EF{}, we show that it nevertheless provides a nontrivial ex-ante guarantee. The proof of this claim can be found in \Cref{apdx:lexicographic}.

\begin{restatable}[\textbf{Guarantee for \UnifPerm{}}]{theorem}{UniformPermutation}
Given as input any instance with lexicographic valuations, the \textnormal{\UnifPerm{}} algorithm computes a randomized allocation that is ex-ante $\frac{1}{2}$-\EF{} and ex-post \EFX{} and \PO{}. %
\label{thm:Uniform_Permutation}
\end{restatable}

\subsubsection*{The ``Eating'' Algorithm}

As mentioned earlier, it was shown in \cite{AFS+24best} that a randomized allocation that is ex-ante \SDEF{} and ex-post \EF{1} can be achieved by Probabilistic Serial~\cite{BM01new}. It is reasonably straightforward to show that the fractional allocation constructed by PS can be decomposed as a probability distribution over allocations, each of which is the outcome of a picking sequence (see \Cref{subsec:Proof_sdEF_EF1+PO_lexicographic} for more details). For lexicographic preferences, any allocation generated by a picking sequence is known to be Pareto optimal \cite{ABL+19efficient,HSV+21fair}. Thus, the existence of an ex-ante \SDEF{} and ex-post $\EF{1}+\PO{}$ randomized allocation follows.

\begin{restatable}[\textbf{Ex-ante \SDEF{} and ex-post \EF{1}+\PO{}}]{proposition}{SDEFandEFOnePOLexicographic}
There is a polynomial-time algorithm that, given as input any instance with lexicographic valuations, returns a randomized allocation that is ex-ante sd-envy-free $(\SDEF{})$ and ex-post \EF{1} and \PO{}.
\label{prop:sdEF_EF1+PO_lexicographic}
\end{restatable}

The ex-post fairness guarantee in \Cref{prop:sdEF_EF1+PO_lexicographic} cannot be strengthened to \EFX{} since, as shown in \Cref{thm:SDEF_EFX}, ex-ante \SDEF{} is incompatible with ex-post \EFX{}. Intuitively, the failure of ex-post \EFX{} occurs because the eating algorithm distributes \emph{all} goods among the agents, and the total fractional amount of goods received by each agent remains equal. When this fractional allocation is decomposed into integral allocations, an envied agent may receive multiple goods, which violates the \EFX{} property for lexicographic preferences (\Cref{prop:EFX_characterization}). 

\subsection{Warm-up: Ex-Ante 6/7-\EF{} and Ex-Post \EFX{}+\PO{}}
\label{subsec:6/7-EF}

In this section, we present an algorithm that outputs a randomized allocation which is ex-ante~$6/7$-\EF{} and ex-post $\EFX{}+\PO{}$. 
We will assume throughout that the number of goods strictly exceeds the number of agents, i.e., $m > n$. Indeed, when $m \leq n$, running the eating algorithm for $m/n$ units of time and subsequently applying the Birkhoff–von Neumann decomposition directly yields an allocation that is ex-ante \EF{} and ex-post \EFX{} + \PO{}.

\begin{algorithm}[t]
    \begin{spacing}{1}
    \DontPrintSemicolon
    \SetAlgoNlRelativeSize{-2} \SetKwInput{KwData}{Input}
    \SetKwInput{KwResult}{Output}
    \KwData{An instance $\langle \, N, M, \mathcal{V} \,\rangle$ with lexicographic valuations}
    \KwResult{A randomized allocation (i.e., a distribution) $\Z$}
    \BlankLine
    \BlankLine
    \tcc{\textbf{Phase I: Simultaneous Eating and B-vN Decomposition}}
    \hrule height 1.2pt
    \BlankLine
    $X \gets$ fractional allocation obtained by running the eating algorithm for one unit of time\;
    $g_i \gets$ the last consumed good by agent $i$ during the eating algorithm\;
    $L \gets \bigcup_{i \in N} g_i$ \tcp*{the set of last consumed goods}
    $(p_\ell, A^\ell)_{\ell \in [q]} \gets \textnormal{B-vN Decomposition}(X)$;
    \BlankLine
    \BlankLine
    \tcc{\textbf{Phase II: Allocate the tail}}
    \hrule height 1.2pt
    \BlankLine
    $\Z \gets $ initialize an empty randomized allocation\;
    \tcc{iterate over allocations in the B-vN support of $X$}
    \For{$\ell \in [q]$}{
        $B \gets A^\ell$\;
        $T \gets$ set of unallocated goods in $B$ \tcp*{the unallocated ``tail'' of goods}
        $N_{L} \gets \{ i \in N : B_i \cap L \neq \emptyset\}$ \tcp*[centre]{set of agents whose bundle in $B$ contains some last consumed good}
    \BlankLine
        \For{each $i \in N_L$}{
            $B \gets (A^\ell_1,\dots,A^\ell_{i-1},A^\ell_i \cup T,A^\ell_{i+1},\dots,A^\ell_n)$ \tcp*[centre]{assign the tail to $i$}
            $\Z \gets \Z \cup \{(p_\ell /|N_L|,B)\}$ \tcp*[centre]{assign weight $p_\ell / |N_L|$ to the complete allocation $B$ in the randomized allocation $\Z$}
        }
    }
    
    \Return{$\Z$}
    \end{spacing}
    \caption{Uniformly Tailed Simultaneous Eating}
    \label{alg:utse}
\end{algorithm}

\paragraph{Description of our algorithm.} Algorithm~\ref{alg:utse} consists of two phases. 

In Phase I, we run the eating algorithm for one unit of time to create an initial fractional allocation $X$ (assuming all agents eat with the same speed of one good per unit time). 
Each agent $i$ can consume several goods over the course of the eating algorithm. For each agent $i \in N$, we define $g_i$ as the \emph{last} (or most recent) good consumed by $i$ during the eating algorithm. We define the set $L$ as the union of all the $g_i$ goods i.e., $L = \bigcup_{i \in N} g_i$. We refer to the set $L$ as the set of {\em last consumed goods}.
 Let $U \coloneqq \{g : \sum_i X_{ig} = 0\}$ denote the set of untouched goods after running the eating algorithm for one unit of time. Note that all goods
outside $L \cup U$ are exhausted strictly before the end of the eating process. Define the \emph{last-consumed mass}, denoted by $k$, as the sum of fractions of goods in $L$ 
consumed by all agents during the eating process. That is, $k \coloneqq \sum_i X_{i,g_i} = \sum_{g \in L} \sum_i X_{ig} =n - |M \setminus (L\cup U)|$. It is easy to see that 
$k \in \{1,\ldots,n\}$.
Then we compute the Birkhoff-von Neumann decomposition $\{(p_\ell, A^\ell)\}_{\ell \in [q]}$ of the partial fractional allocation~$X$~(\Cref{prop:BvN_general}). From the above assumption that $m > n$, it follows that each agent receives exactly one good in each partial allocation $A^\ell$.

In Phase II, our algorithm iterates over the partial allocations $A^1,\dots,A^q$ in the B-vN support of $X$, and turns them into complete allocations through a `tail assignment' step. Specifically, consider an arbitrary partial allocation $Z \in \{A^1,\dots,A^q\}$ in the B-vN support of $X$. Let $T$ denote the unallocated goods in $Z$. We refer to this set as the \emph{tail} to emphasize that, due to lexicographic preferences, the goods in $T$ are considered less valuable by any agent than the single good assigned in Phase I.

The algorithm considers the set of agents $N_{L} \coloneqq \{i \in N : Z_i \cap L \neq \emptyset\}$ whose bundles in $Z$ contain some last consumed good in $L$. We will show in \Cref{lem:nl-size} that $N_L$ is always non-empty. From the set $N_L$, the algorithm picks an agent $i$ uniformly at random and assigns to it the \emph{entire} tail of goods $T$. The resulting randomized allocation $\Z$ is returned as the output of the algorithm.

The intuition behind the tail allocation step is as follows: Recall that at the end of Phase I, each agent receives exactly one good in each partial allocation $A^\ell$. Fix an agent $i \in N$. Among all the goods consumed by agent $i$ during the eating algorithm, that is, the goods in $\{j \in M: X_{i,j} > 0\}$, the last consumed good $g_i$ must be agent $i$'s least-preferred good. Therefore, if agent $i$ receives a good from the set $L$ under $A^\ell$---in other words, $i \in N_L$---it might envy other agents who receive its more-preferred consumed goods in $M$. However, agent $i$ does not envy any other agent who receives a good from $L$. This is because the good $g_i$ is agent $i$'s favorite good in $L$. (This, in turn, is because if there is another good $g' \in L$ such that $g' \>_i g_i$, then agent $i$ would have finished eating $g'$ before switching to $g_i$. However, that would imply that $g' \notin L$ since $g'$ is fully consumed strictly before the end of the algorithm, and all other agents must be consuming other goods at the end of the algorithm.)

Recall from \Cref{prop:EFX_characterization} that under lexicographic preferences, an allocation is \EFX{} if and only if the bundle of each envied agent is a singleton. Therefore, to achieve ex-post \EFX{}, our algorithm assigns the set of unallocated goods $T$ to an unenvied agent in $N_L$ chosen uniformly at random. Note that although the output of our algorithm is a probability distribution, the algorithm itself is deterministic and runs in polynomial time (using \cite[Theorem~2]{AFS+24best}).

We will prove the following result in \Cref{sec:6-over-7-EF}.
\begin{restatable}{theorem}{UtseUpperBound}
Algorithm~\ref{alg:utse} returns a randomized allocation that is ex-ante $\frac{3k}{3k+1}$-\EF{}, where $k \ge 1$ is the last consumed mass of the input instance.
\label{thm:utse-upper-bounnd}
\end{restatable}

Furthermore, we will show that when $k=1$, Algorithm~\ref{alg:utse} is ex-ante \EF{}.
We also have an example (see \Cref{ex:utse-tight}) that shows our analysis for Algorithm~\ref{alg:utse} is tight.

\subsection{Dependent Rounding Algorithm: Ex-Ante 9/10-\EF{} and Ex-Post \EFX{}+\PO{}}
To improve the ex-ante \EF{} guarantee from $6/7$ to $9/10$, observe that we need to only improve our guarantee when $k = 2$. For all $k \ne 2$, Algorithm \ref{alg:utse} outputs a randomized allocation that is at least ex-ante $9/10$-\EF{}. For $k = 2$, we achieve this improvement using dependent rounding. Before we describe our new algorithm, we first describe the key result of dependent rounding, rephrased to fit our context. 

\begin{theorem}[Dependent Rounding \cite{GKSADependent}]\label{thm:dependent-rounding}
Let $X$ be a fractional allocation such that $X_{i, g} \in [0, 1]$ and each good $g$ is allocated an integer amount, $\sum_{i \in N} X_{i, g} \in \mathbb{Z}$. There is a polynomial time algorithm that generates an integral allocation $Z$ from $X$ such that
\begin{enumerate}
    \item[(a)] For each item $g$ and agent $i \in N$, $\Pr[g \in Z_i] = X_{i, g}$
    \item[(b)] For each item $g$, the number of agents who receive $g$ in $Z$ is exactly $\sum_{i \in N} X_{i, g}$ 
    \item[(c)] For each item $g$ and agents $i, j \in N$, $\Pr[g \in Z_i \text{ and } g \in Z_j] \le \Pr[g \in Z_i] \Pr[g \in Z_j]$
\end{enumerate}
\end{theorem}

\begin{algorithm}[t]\begin{spacing}{1}
    \DontPrintSemicolon
    \SetAlgoNlRelativeSize{-2} \SetKwInput{KwData}{Input}
    \SetKwInput{KwResult}{Output}
    \KwData{An instance $\langle \, N, M, \mathcal{V} \,\rangle$ with lexicographic valuations and last consumed }
    \KwResult{(Sampled) Integral allocation $Z$}
    \BlankLine
    \BlankLine
    \tcc{\textbf{Phase 1: Simultaneous Eating and Dependent Rounding}}
    \hrule height 1.2pt
    \BlankLine
    $X \gets$ fractional allocation obtained by running the eating algorithm for one time unit\\
    $g_i \gets$ last consumed good by agent $i$ during Step 1 \\
    $L \gets \bigcup_{i\in N} g_i$\\
    Create a super-good $s$ \\
    $X_{i, s} = X_{i, g_i}$ for all $i \in N$ \\
    Remove all goods in $L \cup U$ \tcp*{All goods in $L \cup U$ are replaced with $s$}
    $Z \gets$ integral allocation sampled from $X$ using `Dependent Rounding' 
    \BlankLine
    
    \tcc{\textbf{Phase 2: Allocating the remaining items}}
    \hrule height 1.2pt
    \BlankLine
    $i, j \gets$ agents who receive `$s$' in $Z$\tcp*{Since $k = 2$, \emph{exactly} two agents receive the super-good} $(a, b) \gets$ uniform random permutation of $(i,j)$ \\
    $Z_a \gets \{g_a\}$\\
    $Z_{b} \gets L \cup U \setminus \{g_a\}$\\

    \Return{$Z$}
    \end{spacing}
    \caption{Dependent Rounding based Algorithm for $k = 2$}
    \label{alg:dependent-rounding}
\end{algorithm}

Note that both (a) and (b) are satisfied by the B-vN decomposition as well. The property that separates dependent rounding from B-vN decomposition is (c), which states that the probability that two agents receive the same item is negatively correlated. Indeed, point (c) is implied by (b) when the total amount a good is allocated is $1$. However, as we describe below, we exploit this property by creating an artificial item (called a super-good) whose total allocation amount is $2$. 

Our algorithm works as follows: we run the eating algorithm for one unit of time to create a fractional allocation $X$, similar to Algorithm \ref{alg:utse}. In this fractional allocation, we replace the set of last consumed goods $L$ with a single {\em super-good} $s$ such that each agent consumes the super-good the same amount as its respective last good. The total amount this super-good is allocated is equal to the \emph{last consumed mass} $k = 2$. We round this new fractional allocation (with the super-good) into an integral allocation $Z$ using dependent rounding (\Cref{thm:dependent-rounding}).

Since $k = 2$, there are two agents who receive the super-good in $Z$ (according to Property (b)). Moreover, all items other than $L \cup U$ are allocated in $Z$. To allocate the items in $L \cup U$, we uniformly at random pick one of the two agents who receive the super-good; say we pick agent $i$ and not agent $j$. We give $i$ its favorite item $g_i$ in $L \cup U$, and give $j$ all the remaining items $L \cup U \setminus \{g_i\}$.

We remark that although this algorithm generates a randomized allocation, the support of this randomized allocation may be exponential (unlike Algorithm \ref{alg:utse}). However, since sampling an integral allocation can be done in polynomial time, the pseudocode of Algorithm \ref{alg:dependent-rounding} outputs a randomized integral allocation instead of the complete distribution over integral allocations.

Following the same arguments of \Cref{prop:ex-post utse}, Algorithm \ref{alg:dependent-rounding} always outputs an EFX+PO allocation. We therefore only focus on the ex-ante \EF{} guarantee.  
To provide some intuition on why dependent rounding may lead to a better result, consider the proof of \Cref{thm:utse-upper-bounnd}. In the proof, we lower bound $\E[v_i(Z_i)]$ using $v_i(X_i)$; that is, we assume in the worst-case that $i$ does not receive any utility from the goods it receives in Phase 2. As \Cref{ex:utse-tight} shows, we cannot strengthen this assumption using the B-vN decomposition alone since it could be that some item in $L$ (say $g$) that provides high utility to $i$ is always allocated when $i$ receives $g_i$. Therefore $g$ is never part of the tail that $i$ receives. 

Using dependent rounding with the super-good $s$, the probability that agent $i$ is allocated the super-good is negatively correlated with the probability that any agent whose last consumed good is $g$ is allocated the super-good. Therefore, there is a higher probability that agent $i$ receives $g$ in Phase 2. Taking this into account leads to a stronger lower bound of $\E[v_i(Z_i)]$. However, this is not straightforward since there is now a chance that agent $i$ could be allocated the super-good together with another agent whose last consumed good is $g_i$. In such a case, agent $i$ receives $g_i$ with only half probability. So there could be some cases where $\E[v_i(Z_i)]$ is {\em less than} $v_i(X_i)$. However, again, the probability of this bad event happening can be upper bounded using the negative correlation property of dependent rounding. Following this intuition, using some careful calculations, we prove that Algorithm \ref{alg:dependent-rounding} outputs a randomized allocation that is ex-ante $\nicefrac{9}{10}$-\EF{}. 

We will prove the following result in \Cref{sec:main-algo}.

\begin{restatable}{theorem}{KTwo}
When the last consumed mass $k = 2$, Algorithm \ref{alg:dependent-rounding} outputs an allocation that is ex-ante $\frac{9}{10}$-\EF{} and ex-post \EFX{} and \PO{}.
\label{thm:k-two}
\end{restatable}

\subsection{An Asymptotically Optimal Result}\label{sec:1-o(1)}
In this section, we present an algorithm that outputs a randomized allocation that achieves ex-ante $(1-\frac{2}{n+1})$-\EF{} and ex-post $\EFX{}+\PO{}$ where $n$ is the number of agents in the instance. Thus, as the instance size grows, Algorithm \ref{alg:1-o(1)} returns an ex-ante guarantee better than $\nicefrac{9}{10}$, and achieves exact ex-ante \EF{} in the limit.   

\begin{algorithm}[h]
    \begin{spacing}{1}
    \DontPrintSemicolon
    \SetAlgoNlRelativeSize{-2} \SetKwInput{KwData}{Input}
    \SetKwInput{KwResult}{Output}
    \KwData{An instance $\langle \, N, M, \mathcal{V} \,\rangle$ with lexicographic valuations}
    \KwResult{A randomized allocation (i.e., a distribution) $\Z$}
    \BlankLine
    \BlankLine
    \tcc{\textbf{Phase I: Simultaneous Eating with Dummy Good}}
    \hrule height 1.2pt
    \BlankLine
    Add a dummy good $d$ to $M$ such that every agent prefers $d$ the most \; 
    $X \gets$ fractional allocation obtained by running the eating algorithm for one unit of time\;
    $(p_\ell, A^\ell)_{\ell \in [q]} \gets \textnormal{B-vN Decomposition}(X)$\;
    \BlankLine
    \tcc{\textbf{Phase II: Allocating the tail}}
    \hrule height 1.2pt
    \BlankLine
    $\Z \gets $ initialize an empty randomized allocation\;
    \tcc{iterate over allocations in the B-vN support of $X$}
    \For{$\ell \in [q]$}{
        $i \gets$ the agent who receives the dummy good $d$ in $A^{\ell}_i$\;
        $T \gets$ set of unallocated goods in $A^{\ell}$ \tcp*{unallocated `tail' of goods}
        $A^{\ell}_i \gets T$\;
        $\Z \gets \Z \cup \{(p_\ell, \,A^{\ell})\}$ \tcp*[centre]{assign weight $p_\ell$ to the now complete allocation $A^{\ell}$ in the randomized allocation $\Z$}
    }
    \Return{$\Z$}
    \end{spacing}
    \caption{Asymptotic Ex-ante EF with Lexicographic Valuations}
    \label{alg:1-o(1)}
\end{algorithm}

\paragraph{Description of Algorithm \ref{alg:1-o(1)}} Once again, the algorithm works in two phases. In Phase I, we run the eating algorithm for $\frac{n-1}{n}$ time units. Let $X$ be the corresponding fractional allocation obtained. Then we compute the B-vN decomposition of the partial fractional allocation $X$. This is done in a slightly modified way to ensure that in every integral allocation appearing in this decomposition, exactly one agent receives no goods. 
Specifically, we augment the instance with a dummy good $d$ that every agent ranks first. Running the eating algorithm for one unit of time on this augmented instance yields a decomposition in which each integral allocation assigns the dummy good to exactly one agent. That agent, in turn, corresponds to the one who receives nothing in the decomposition of the original partial allocation $X$.

In Phase II, we assign all remaining items to this particular agent. Similar to our previous algorithms, it is straightforward to see that Algorithm \ref{alg:1-o(1)} can also be implemented in polynomial time. 

\begin{restatable}{theorem}{thma}
Algorithm \ref{alg:1-o(1)} returns a randomized allocation that is ex-ante $(1-\frac{2}{n+1})$-\EF{} and ex-post $\EFX{}+\PO{}$, where $n$ is the number of agents in the input instance.
\label{thm:eatless}
\end{restatable}

\begin{proof}
    It suffices to show the ex-ante guarantee as the ex-post guarantee follows from a similar argument as in 
    the proof of Proposition~\ref{prop:ex-post utse}. 
    
    Let us fix two agents $i,j$, and analyze $i$'s envy towards $j$. Let $X$ denote fractional allocation obtained after Phase I. Clearly,  
    \begin{align}
        \E[v_i(X_i)] \geq \E[v_i(X_j)]  \label{eqn1:1-o(1)} 
    \end{align}

Thus, only the `tail' phase could contribute to $i$'s ex-ante envy towards $j$. Note that the marginal probability that any agent receives additional goods in Phase II is exactly $\frac{1}{n}$. To see this, let us revisit the idea of the equivalent instance with the dummy good $d$. Every agent consumes exactly $\frac{1}{n}$ fraction of $d$, and thus receives the tail with probability $\frac{1}{n}$. Let $T_j$ denote the tail that $j$ receives and $g_i$ be the last good that agent $i$ ate from in Phase I. Note that from agent $i$'s perspective, $T_j$ can only contain $g_i$ and items worse than $g_i$. Hence, by lexicographicity,    
    \begin{align}
        v_i(T_j) \leq 2v_i(g_i). \label{eqn2:1-o(1)}
    \end{align}
    
    From equations \eqref{eqn1:1-o(1)} and \eqref{eqn2:1-o(1)} it is easy to see that,
    
    \begin{align*}
        \E[v_i(\Z_j)]-\E[v_i(\Z_i)] \leq \frac{1}{n}\cdot v_i(T_j) \leq \frac{2v_i(g_i)}{n} 
        &\leq \frac{2\cdot\E[v_i(X_i)]}{n-1} \leq \frac{2\cdot\E[v_i(\Z_i)]}{n-1}
    \end{align*}

    Rewriting this we get,
    \begin{equation*}
        \E[v_i(\Z_i)] \geq \frac{n-1}{n+1}\cdot \E[v_i(\Z_j)] = \left ( 1 - \frac2{n+1}\right ) \E[v_i(\Z_j)]. \qedhere
    \end{equation*}
    \qedhere
\end{proof}

Note that Algorithm \ref{alg:1-o(1)} returns a better ex-ante guarantee asymptotically compared to the Dependent Rounding technique (Algorithm \ref{alg:dependent-rounding}). However, for instances with $n\leq 19$, we still have $1-\frac{2}{n+1} \leq \frac{9}{10}$. Thus, running either Algorithm \ref{alg:utse} or Algorithm \ref{alg:dependent-rounding} (depending on the value of $k$) along with Algorithm \ref{alg:1-o(1)} and taking the maximum of both guarantees, gives us our main result.  

\dependentrounding*

\begin{figure}[H]
    \centering
\includegraphics[width=0.45\linewidth]{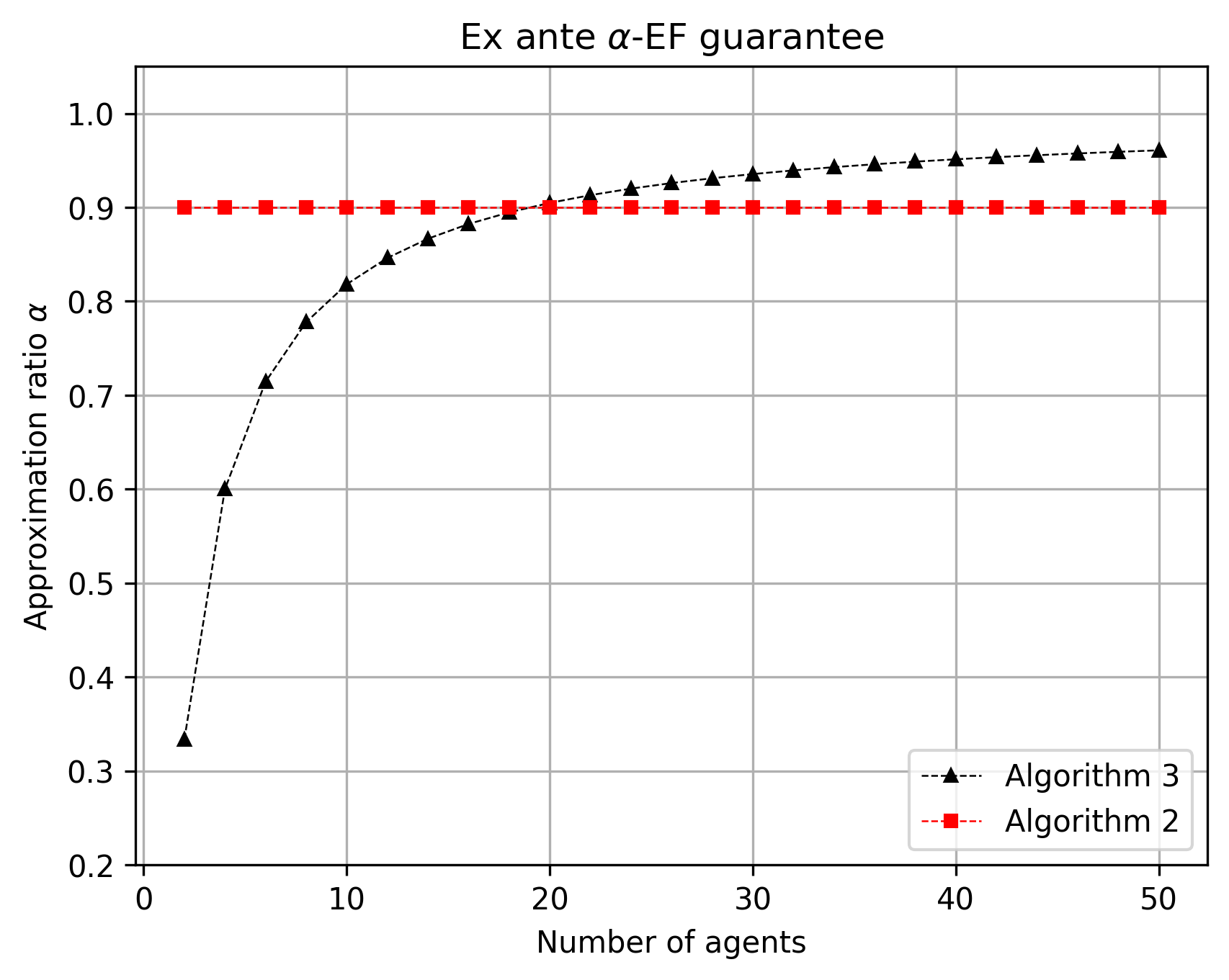}
    \caption{Approximation guarantees of Algorithms \ref{alg:dependent-rounding} and \ref{alg:1-o(1)} as a function of the size of the instance}
    \label{fig:plot}
\end{figure}

To further show the contrast between Algorithms \ref{alg:dependent-rounding} and \ref{alg:1-o(1)}, we show \Cref{app:1-o(1) example} in \Cref{apdx:1-o(1)}, an instance where Algorithm \ref{alg:1-o(1)} cannot perform better than ex-ante $\frac{4}{5}$-\EF{}, while the Dependent Rounding Algorithm still preserves the guarantee of ex-ante $\frac{9}{10}$-\EF{}.

\section{Results for Subadditive and Monotone Valuations}
\label{sec:Results-Additive}

This section proves our best-of-both-worlds results for monotone valuations (Theorem~\ref{thm:bobw-charity-monotone}) and subadditive valuations (Theorem~\ref{thm:Ex-ante_half-Prop_ex-post_EFX-with-charity}).

\subsection{Ex-ante 1/2-EF and ex-post EFX-with-charity}

In this subsection, we prove the following best-of-both-worlds guarantee for \emph{monotone} valuations:

\BoBWCharityMonotone*

Our algorithm (Algorithm~\ref{alg:random-charity-swap}) begins with the empty allocation, and initializes a pool $P = M$ of unallocated goods called the {\em charity}. Then, while there exists an agent that envies the un-allocated pool $P$, the algorithm chooses an arbitrary \textit{minimal envied subset} (a subset that is envied by at least one agent, but no agent envies any strict subset of it), say $Q$ of the charity. Then, it picks an agent uniformly at random from the set of agents who envy $Q$, and swaps its bundle with $Q$ (giving its original bundle back to the un-allocated pool).

\begin{algorithm}[h]\begin{spacing}{1}
        \SetAlgoNlRelativeSize{-2} \DontPrintSemicolon
        \SetKwInput{KwData}{Input}
        \SetKwInput{KwResult}{Output}
        \KwData{An instance $\langle N,M,\V\rangle$ of the fair division problem.}
        \KwResult{A random allocation $R$.}
        $P \gets M$\;
        $R_i \gets \emptyset$ for all $i \in N$\;
        \While{$\exists h \in N$  with $v_h(R_h) < v_h(P)$}{
            $Q \gets $ an arbitrary inclusion-wise minimal envied subset of $P$\;
            $H \gets \{i \in N: v_i(X_i) < v_i(Q)\}$\;
            $t \gets $ an agent chosen {\it uniformly at random} from $H$\;
            $P \gets (P \setminus Q) \cup R_t$.\\
            $R_t \gets Q$\;
        }
        \Return{R}
    \end{spacing}
    \caption{Ex-ante 1/2-\EF{} and Ex-post \EFXcharity{} }
    \label{alg:random-charity-swap}
\end{algorithm}

When $k_r$ is picked deterministically, our algorithm reduces to the deterministic \EFXUnenviedCharity{} algorithm of \cite{CKM+21little}. Since their algorithm runs in pseudo-polynomial time, Algorithm~\ref{alg:random-charity-swap} also runs in pseudo-polynomial time and returns an ex-post \EFXUnenviedCharity{} allocation $X$. The key insight in proving the ex-ante guarantee is that whenever agent $j$ gets a bundle that agent $i$ envies, agent $i$ would have had an ``equal opportunity'' of getting that bundle. We will formalize this insight to prove ex-ante $\nicefrac{1}{2}$-\EF{} via approximate stochastic dominance. In particular, consider any two agents $i$ and $j$, and any threshold value $T$. Then, the following lemma proves \Cref{thm:bobw-charity-monotone}:

\begin{lemma}
\label{lem:sd_random_charity}
For every $T \geq 0$, $\Pr[v_i(R_i) \geq T] \geq \frac{1}{2} \cdot \Pr[v_i(R_j)\geq T]$.
\end{lemma}

\begin{proof}
 Let us call the $r^{\text{th}}$ iteration of the while loop {\it significant} if the following properties hold true:
\begin{inparaenum}[(a)]
    \item[(a)] $v_i(Q_r) \geq T$,
    \item[(b)] $i, j \in H_r$, and
    \item[(c)] $t_r \in \{i, j\}$.
\end{inparaenum}

Let $E$ denote the event that there exists a significant iteration in the run of the algorithm, and let $\bar{E}$ be its complement, that is, the event that no iteration is significant.

Let us first consider the event $E$. Let $r_0$ be the first significant iteration. Then,
\begin{equation} \label{eqn:event_1}
    \Pr[v_i(R_i) \geq T \mid  E] \geq \Pr[t_{r_0} = i \mid E] = \frac{1}{2} \geq \frac{1}{2} \cdot \Pr[v_i(R_j) \geq T \mid  E],
\end{equation}

where the first inequality holds true because the agent utilities are non-decreasing over the course of the algorithm, and $\Pr[t_{r_0} = i \mid E] = \frac{1}{2}$ because, given $r_0$ is the first significant iteration, $t_{r_0} \in \{i, j\}$, and $t_{r_0} = i$ and $t_{r_0} = j$ are equally likely.

Now, consider the event $\bar{E}$. We will show that, under $\bar{E}$, $v_i(R_j) \geq T \implies v_i(R_i) \geq T$. Indeed suppose $v_i(R_j) \geq T$ and consider the iteration $r$ where agent $j$ gets $Q_r = R_j$. Given $\bar{E}$, this iteration must not be significant, which can only happen if $i \notin H_r$. So, at that instant the agent $i$ does not envy $Q_r = R_j$. Given that agent utilities only increase as the algorithm runs, $v_i(R_i) \geq v_i(Q_r) \geq T$. Therefore:
\begin{equation} \label{eqn:event_2}
    \Pr[v_i(R_i) \geq T \mid \bar{E}] \geq \Pr[v_i(R_j) \geq T \mid \bar{E}].        
\end{equation}

\Cref{eqn:event_1,eqn:event_2} together finish the proof.
\end{proof}

\subsection{Ex-ante 1/3-\EF{} and ex-post \EFXBoundedCharity{}}

In this subsection, we prove \Cref{thm:Ex-ante_half-Prop_ex-post_EFX-with-charity} and \Cref{thm:random_source_algo} for subadditive valuations.

\PSandCharity*

Our algorithm starts with the partial \EFX{} allocation $R$ returned by Algorithm~\ref{alg:random-charity-swap}, and computes an \EFXBoundedCharity{} allocation $R'$ by using the little charity algorithm.

\begin{algorithm}[h]\begin{spacing}{1.2}
        \SetAlgoNlRelativeSize{-2} \DontPrintSemicolon
        \SetKwInput{KwData}{Input}
        \SetKwInput{KwResult}{Output}
        \KwData{An instance $\langle N,M,\V\rangle$ of the fair division problem.}
        \KwResult{A random allocation $R'$.}
        \BlankLine
        $R \gets $ output of Algorithm~\ref{alg:random-charity-swap} run on $\langle N,M,\V\rangle$. \;
        Starting from $R$, compute an \EFXBoundedCharity{} allocation $R'$ by using the little charity algorithm~\cite{CKM+21little}.\;
        \Return{$R'$}
    \end{spacing}
    \caption{ex-post \EFXBoundedCharity{} + ex-ante 1/2-\Prop{}}
    \label{alg:little-charity-bobw}
\end{algorithm}

Clearly, Algorithm~\ref{alg:little-charity-bobw} returns an allocation that is ex-post \EFXBoundedCharity{}. Therefore, the following lemma proves \Cref{thm:Ex-ante_half-Prop_ex-post_EFX-with-charity}:

\begin{lemma}
    The output of Algorithm~\ref{alg:little-charity-bobw} is ex-ante $\nicefrac{1}{2}$-\Prop{}.
\end{lemma}

    \begin{proof}
    Let us fix an agent $i \in N$. Since $R$ is ex-ante $\nicefrac{1}{2}$-\EF{}, and no agent envies the pool $P$ of un-allocated items in $R$, we have $v_i(P) \leq v_i(R_i)$, and:

    \[
    \E[v_i(R_j)] \leq 2 \cdot \E[v_i(R_i)] \text{ for all } j \neq i.
    \]
    
    Also, since the agent utilities are non-decreasing, we have $v_i(R_i) \leq v_i(R'_i)$. Hence,
    \begin{align*}
        v_i(M) = \E[v_i(M)] &= \E[v_i(P \cup_{j \in N}R_j)] \\
        &\leq \E[v_i(P) + v_i(R_i) + \sum_{j \neq i} v_i(R_j)] \qquad \text{(using subadditivity)}\\
        &= \E[v_i(P)] + \E[v_i(R_i)] + \sum_{j \neq i} \E[v_i(R_j)]\\
        &\leq \E[v_i(R_i)] + \E[v_i(R_i)] + 2 \cdot (n-1) \cdot \E[v_i(R_i)]\\
        &= 2n \cdot \E[v_i(R_i)]\\
        &\leq 2n \cdot \E[v_i(R'_i)].
    \end{align*}\qedhere
\end{proof}

Recall that the output $R$ of Algorithm~\ref{alg:random-charity-swap} is ex-post \EFXcharity{}. Suppose that there exists an agent $s$ in $R$ that no agent envies. We call such an agent a \emph{source} in the envy graph. Since no agent envies charity, we could allocate the items in charity to agent $s$ to get a complete allocation that is ex-post $\nicefrac{1}{2}$-\EFX{}. Additionally, note that in this new allocation, the \EFX{} condition is satisfied for any pair of agents not including $s$. To ensure the existence of a source in the envy graph of the resulting \EFXcharity{} allocation, we begin with selecting an agent 
$s$ uniformly at random and maintain $s$ as a source throughout the algorithm. This is achieved by ensuring that agent $s$ never receives a bundle that is envied by another agent. Hence, whenever $s$ is the unique agent who envies the current minimal envied subset, we allocate that subset to $s$; otherwise, we allocate it to an agent other than $s$ who envies the subset (chosen u.a.r). Informally, the ex-ante guarantee follows from the observation that, whenever an agent is not the source agent, it does not envy other agents too much in expectation because whenever it envies the minimal envied subset, it gets a fair shot at receiving the same. The source agent can end up envying other agents quite a bit, but the probability of an agent being the source is small: $\nicefrac{1}{n}$.

\begin{algorithm}[t]\begin{spacing}{1}
        \SetAlgoNlRelativeSize{-2} \DontPrintSemicolon
        \SetKwInput{KwData}{Input}
        \SetKwInput{KwResult}{Output}
        \KwData{A fair division instance $\langle N,M,\V\rangle$ with subadditive valuations}
        \KwResult{A random allocation $R'$.}
        $P \gets M$; $R_i \gets \emptyset \,\forall\, i \in N$\;
        $s \gets$ an agent chosen uniformly at random from $N$
        
        \While{$\exists h \in N$  with $v_h(R_h) < v_h(P)$}{
            $Q \gets $ an arbitrary inclusion-wise minimal envied subset of $P$\;
            $H \gets \{i \in N: v_i(R_i) < v_i(Q)\}$\;
            \uIf{$H = \{s\}$}{
                $t \gets s$} \Else{
                $t \gets $ an agent chosen {\it uniformly at random} from $H \setminus \{s\}$\;
            }
            $P \gets (P \setminus Q) \cup R_t$\;
            $R_t \gets Q$\;
        } \label{l:efx_charity_end}
        $R' \gets R$\;
        $R'_{s} \gets R'_{s} \cup P$\;
        \Return{$R'$}
    \end{spacing}
    \caption{Ex-ante 1/3-\EF{} and Ex-post 1/2-\EFX{} for Subadd. Vals.}
    \label{alg:random-charity-swap-random-source}
\end{algorithm}

\randomsource*

\begin{proof}
We refer to Algorithm \ref{alg:random-charity-swap-random-source}.
The partial allocation $R$ at the end of the while loop (line~\ref{l:efx_charity_end}) clearly satisfies $\EFX{}$ since we always allocate a minimal envied subset, and by definition, no strict subset of this bundle is envied. Additionally, no agent envies the agent $s$ in $R$, since whenever $s$ is allocated a bundle, no other agent envies it. The output $R'$ of the algorithm differs from $R$ only in the bundle of agent $s$. The bundle $R'_{s}$ allocated to agent $s$ is a union of two subsets ($R_{s}$ and the charity) that are not envied by any other agent. Hence, subadditivity implies that the algorithm always returns an ex-post $\nicefrac{1}{2}$-\EFX{} allocation. Let us now analyze the ex-ante guarantee. \\

Consider any two agents $i \ne j \in N$. First, consider the event that $i$ is not the source agent; that is $s \neq i$. When $s = j$, $v_i(R'_j) = v_i(R_j \cup P) \leq v_i(R_j) + v_i(P) \leq v_i(R_i) + v_i(R_i) \le 2 v_i(R_i) \leq 2 v_i(R'_i)$. When $s \notin \{i,j\}$, $R'_i = R_i$ and $R'_j = R_j$. Hence, by the approximate stochastic domination proved in Lemma~\ref{lem:sd_random_charity}, $\E[v_i(R'_j) \mid s \notin \{i,j\}] \leq 2 \E[v_i(R'_i) \mid s \notin \{i,j\}]$. Overall, 
$$\E[v_i(R'_j) \mid s \neq i] \leq 2 \E[v_i(R'_i) \mid s \neq i]$$

Now, let $g$ be the favorite good of agent $i$. There must be an iteration $r$ of the while loop where $i$ envies $Q_r$ and $v_i(Q_r) \geq v_i(g)$; it cannot be the case that $g$ is unallocated and $v_i(R_i) < v_i(g)$ at the end of the algorithm. 
Let $c$ be the first iteration where $v_i(Q_c) \geq v_i(g)$. Since this is the first such iteration, we must have $i \in H_c$. Thus $i$ receives $Q_c$ with probability $\frac{1}{|H_c \setminus \{s\}|} \geq \frac{1}{n-1}$. Hence,
\begin{align*}
    \E[v_i(R'_i)] \geq \Pr[s \neq i] \cdot\E[v_i(R_i) \mid s \neq i] \geq \frac{n-1}{n} \frac{v_i(g)}{n-1} \geq \frac{v_i(g)}{n}.
\end{align*}

When $i$ is the source agent (that is, $s = i$), since $R$ is \EFX{}, we have that $v_i(R'_i) \geq v_i(R_i) \geq v_i(R_j) - v_i(g) = v_i(R'_j) - v_i(g)$. Hence, $\E[v_i(R'_j) \mid s = i] \leq \E[v_i(R'_i) \mid s = i] + v_i(g)$. We can now upper bound the value of $\E[v_i(R'_j)]$.

\begin{align*}
    \E[v_i(R'_j)] &= \Pr[s \neq i] \cdot \E[v_i(R'_j) \mid s \neq i] + \Pr[s = i] \cdot \E[v_i(R'_j) \mid s = i]\\
    &\leq 2 \Pr[s \neq i] \cdot \E[v_i(R'_i) \mid s \neq i] + \frac{1}{n} \left( \E[v_i(R'_i) \mid s = i] + v_i(g) \right) \\
    &\leq 2 \Pr[s \neq i] \cdot \E[v_i(R'_i) \mid s \neq i] + \frac{1}{n} \left( 2 \E[v_i(R'_i) \mid s = i] + n \E[v_i(R'_i)] \right) \\
    &\leq 2 \left( \Pr[s \neq i] \cdot \E[v_i(R'_i) \mid s \neq i] + \Pr[s = i] \cdot \E[v_i(R'_i) \mid s = i]\right) \\
    & \qquad \qquad + \E[v_i(R'_i)]\\
    &= 3 \E[v_i(R'_i)].
\end{align*}\qedhere
\end{proof}

\section{Concluding Remarks}
We studied best-of-both-worlds fairness guarantees for lexicographic, subadditive, and monotone valuations. Towards our goal of achieving ex-post \EFX{} guarantees, we showed that ex-post \EFX{}+\PO{} can be achieved in poly. time alongside ex-ante $\max\{\frac{9}{10},1-\frac{2}{n+1}\}$-\EF{} under lexicographic valuations. This result involves taking the maximum over three algorithms, one of which involves a novel application of the dependent rounding technique. An exciting open question in this direction is to examine whether the ex-ante fairness guarantee can be improved to \emph{exact} \EF{}. Given the prevalence of lexicographic behavior in practical decision-making processes, we believe these results may have some real-world applications. 

We also studied ex-post \EFXcharity{} for subadditive and monotone valuations, and established the first best-of-both-worlds guarantees in this context. Our algorithms are extremely simple; however, their running times are pseudopolynomial. Note that no polynomial time algorithm is currently known for computing an \EFXcharity{} allocation even without any ex-ante guarantee and even under additive valuations. Moreover, whether the ex-post guarantee for monotone can be strengthened to EFX-with-bounded-charity is an interesting open problem. 

Finally, it is worth highlighting an open problem related to the \SDEF{} notion. Our main negative result showed that ex-ante \SDEF{} and ex-post \EFX{} are incompatible under lexicographic preferences. Since \SDEF{} requires a fractional allocation to be envy-free with respect to all consistent \emph{additive} instances, this impossibility may seem slightly unsatisfactory. An arguably more reasonable notion would require consistency only with respect to all \emph{lexicographic} instances, a subclass of additive valuations. Let us define a fractional allocation to be \emph{lex-\SDEF{}} if it is envy-free with respect to any consistent lexicographic valuations instance. It is interesting to ask if ex-ante lex-\SDEF{} is compatible with ex-post \EFX{}. Additionally, just like \SDEF{} is characterized by a ``prefix-domination'' property, it would be interesting to obtain a similarly succinct characterization for lex-\SDEF{}.

\section*{Acknowledgments}

We thank anonymous reviewers for their helpful comments and feedback. Part of this work was completed during SP's MS thesis project at the Department of Computer Science and Engineering at IIT Delhi, which took place from June 2024 to April 2025. It also includes contributions made during VV's visit to the department in January 2025. We gratefully acknowledge the hospitality and support provided by the CSE Department at IIT Delhi during both periods. We are also thankful to Vishwa Prakash HV for useful comments.

TK acknowledges support from the Department of Atomic Energy, Government of India, under project no. RTI4014. SP acknowledges support from the Department of Science and Technology, Government of India, through the INSPIRE Scholarship for Higher Education (SHE). RV acknowledges support from SERB grant no. CRG/2022/002621, DST INSPIRE grant no. DST/INSPIRE/04/2020/000107, iHub Anubhuti IIITD Foundation, and Mr. D.P. Gupta Chair Professorship. VV acknowledges funding from the National Science Foundation (NSF) Career Award 2441296 and Grant RI-2327057. JY acknowledges support from the Google PhD Fellowship. 

\bibliographystyle{alpha}
\bibliography{References_V2}
 
\clearpage
\appendix
\begin{center}
    \Large{\textbf{Appendix}}
\end{center}

\section{Omitted Proofs from Section \ref{sec:impossible}}\label{apdx:impossible}

\SDEFandEFX*

\begin{proof}
Consider the following instance with four goods $g_1,\dots,g_4$ and three agents $1,2,3$ with lexicographic preferences:
\begin{align*}
1: ~&~ g_1 \> g_3 \> g_4  \> g_2\\
2: ~&~ g_1 \> g_2 \> g_4  \> g_3\\
3: ~&~ g_2 \> g_3 \> g_4  \> g_1
\end{align*}
We use the ordinal representation of the lexicographic preferences in the above instance. Thus, for example, agent 1 prefers any bundle containing good $g_1$ over any bundle that does not, subject to that it prefers any bundle containing $g_3$ over any bundle without it, and so on.

We will first argue, via case analysis, that the above instance has four \EFX{} allocations.

\begin{enumerate}[label=\textbf{Case \Roman*:},align=left]%
    \item When $g_1$ is assigned to agent 1.
    
    Since $g_1$ is the top-ranked good of agent 2, agent 1 cannot receive any other good under any allocation that assigns $g_1$ to agent 1 (otherwise, by \Cref{prop:EFX_characterization}, \EFX{} will be violated). Thus, the good $g_2$ must be assigned to either agent 2 or agent 3, leading to the following subcases.
    
    \begin{enumerate}[label=\arabic*.]
    \item If $g_2$ is assigned to agent 2, then agent 3 will envy it regardless of how the remaining items are assigned. Thus, by \Cref{prop:EFX_characterization}, all remaining goods must be assigned to agent 3. The resulting allocation is $A^1 \coloneqq (\{g_1\},\{g_2\},\{g_3,g_4\})$, i.e., $$\NiceMatrixOptions{code-for-first-row=\scriptstyle,code-for-first-col=\scriptstyle}
A^1{}=\begin{bNiceMatrix}[first-row,first-col]
      & g_1 & g_2 & g_3 & g_4 \\
      1 & 1 & \cdot & \cdot & \cdot \\
      2 & \cdot & 1 & \cdot & \cdot \\
      3 & \cdot & \cdot & 1 & 1 
  \end{bNiceMatrix}$$
  which satisfies \EFX{}.
  
    \item Now suppose $g_2$ is assigned to agent 3. Then, regardless of how the remaining goods are assigned, agent 2 will envy both agents 1 and 3. Thus, in order to maintain \EFX{}, all remaining goods must be assigned to agent 2, resulting in the \EFX{} allocation $A^2 \coloneqq (\{g_1\},\{g_3,g_4\},\{g_2\})$, i.e., $$\NiceMatrixOptions{code-for-first-row=\scriptstyle,code-for-first-col=\scriptstyle}
A^2{}=\begin{bNiceMatrix}[first-row,first-col]
      & g_1 & g_2 & g_3 & g_4 \\
      1 & 1 & \cdot & \cdot & \cdot \\
      2 & \cdot & \cdot & 1 & 1 \\
      3 & \cdot & 1 & \cdot & \cdot 
  \end{bNiceMatrix}.$$
    \end{enumerate}
    
    \item When $g_1$ is assigned to agent 2.
    
    Note that agent 2 cannot receive any other good under any allocation that assigns to it the good $g_1$ (otherwise, by \Cref{prop:EFX_characterization}, \EFX{} will be violated from agent 1's perspective). Thus, the good $g_2$ must be assigned to either agent 1 or agent 3, leading to the following subcases.
    
    \begin{enumerate}[label=\arabic*.]
    \item If $g_2$ is assigned to agent 1, then agent 3 will envy it regardless of how the remaining items are assigned. Thus, by \Cref{prop:EFX_characterization}, the remaining goods $g_3$ and $g_4$ must be assigned to agent 3. However, doing so violates \EFX{} from agent 1's perspective. Thus, there is no \EFX{} allocation in this case.
  
    \item Now suppose $g_2$ is assigned to agent 3. Then, the good $g_3$ cannot be assigned to agent 3 (as that would violate \EFX{} from agent 1's perspective) and must therefore be assigned to agent 1. The remaining good $g_4$ can be allocated to either agent 1 or agent 3, resulting in the following two \EFX{} allocations:
    $$A^3 \coloneqq (\{g_3\},\{g_1\},\{g_2,g_4\}), \text{ i.e. }, \NiceMatrixOptions{code-for-first-row=\scriptstyle,code-for-first-col=\scriptstyle}
A^3{}=\begin{bNiceMatrix}[first-row,first-col]
      & g_1 & g_2 & g_3 & g_4 \\
      1 & \cdot & \cdot & 1 & \cdot \\
      2 & 1 & \cdot & \cdot & \cdot \\
      3 & \cdot & 1 & \cdot & 1 
  \end{bNiceMatrix},$$
    $$A^4 \coloneqq (\{g_3,g_4\},\{g_1\},\{g_2\}), \text{ i.e. }, \NiceMatrixOptions{code-for-first-row=\scriptstyle,code-for-first-col=\scriptstyle}
A^4{}=\begin{bNiceMatrix}[first-row,first-col]
      & g_1 & g_2 & g_3 & g_4 \\
      1 & \cdot & \cdot & 1 & 1 \\
      2 & 1 & \cdot & \cdot & \cdot \\
      3 & \cdot & 1 & \cdot & \cdot 
  \end{bNiceMatrix}.$$
    \end{enumerate}
    
    \item When $g_1$ is assigned to agent 3.
    
    Since $g_1$ is the top-ranked good of agents 1 and 2, agent 3 cannot receive any other good under any allocation that assigns to it the good $g_1$ (otherwise, by \Cref{prop:EFX_characterization}, \EFX{} will be violated). Thus, the good $g_2$ must be assigned to either agent 1 or agent 2, leading to the following subcases.
    
    \begin{enumerate}[label=\arabic*.]
    \item If $g_2$ is assigned to agent 1, then agent 3 will envy it regardless of how the remaining items are assigned. Thus, by \Cref{prop:EFX_characterization}, the remaining goods $g_3$ and $g_4$ must be assigned to agent 2. However, doing so violates \EFX{} from agent 1's perspective. Thus, there is no \EFX{} allocation in this case.
  
    \item Now suppose $g_2$ is assigned to agent 2. Once again, due to the envy from agent 3, the remaining goods $g_3$ and $g_4$ must be allocated to agent 1. However, this violates \EFX{} from agent 3's perspective. Hence, there is no \EFX{} allocation in this case.
    \end{enumerate}
\end{enumerate}

Thus, the four allocations $A^1,\dots,A^4$ listed above are the only ones that are \EFX{} for the given instance.

Now suppose, for contradiction, that there exists a randomized allocation $\X$ that is ex-ante \SDEF{} and ex-post \EFX{}. Then, the associated fractional allocation $X$ can be written as $X = \sum_{i=1}^{4} \alpha_i A^i$ such that $\sum_{i=1}^4 \alpha_i= 1$ and $\alpha_i \geq 0$ for all $i \in [4]$, i.e.,
$$\NiceMatrixOptions{code-for-first-row=\scriptstyle,code-for-first-col=\scriptstyle}
\small
X{}=\begin{bNiceMatrix}[first-row,first-col]
      & g_1 & g_2 & g_3 & g_4 \\
      1 & \alpha_1+\alpha_2 & \cdot & \alpha_3+\alpha_4 & \alpha_4 \\
      2 & \alpha_3+\alpha_4 & \alpha_1 & \alpha_2 & \alpha_2\\
      3 & \cdot & \alpha_2+\alpha_3+\alpha_4 & \alpha_1 & \alpha_1+\alpha_3 
  \end{bNiceMatrix}.$$
  
  We will now write the conditions for sd-envy-freeness for every pair of agents. %
  For the sake of completeness, we mention all the constraints; however it is only the numbered ones that will be used to show the impossibility.
  
  \begin{itemize}
      \item \SDEF{} from agent 1's perspective towards agent 2
      \begin{align}
      g_1 & : & \alpha_1+\alpha_2 & \geq \alpha_3+\alpha_4\label{eqn:1-for-2-top_one}\\
      g_1,g_3 & : & \alpha_1 + \alpha_2 + \alpha_3+\alpha_4 & \geq  \alpha_2 + \alpha_3+\alpha_4 \nonumber\\
      g_1,g_3,g_4 & : & \alpha_1+\alpha_2+\alpha_3+2\alpha_4 & \geq 2\alpha_2+\alpha_3+\alpha_4 \nonumber\\
      g_1,g_3,g_4,g_2 & : & \alpha_1+\alpha_2+\alpha_3+2\alpha_4 & \geq \alpha_1+2\alpha_2+\alpha_3+\alpha_4\label{eqn:1-for-2-top_four}
  \end{align}
      \item \SDEF{} from agent 1's perspective towards agent 3
      \begin{align}
      g_1 & : & \alpha_1+\alpha_2 & \geq 0 \nonumber\\
      g_1,g_3 & : & \alpha_1+\alpha_2+\alpha_3+\alpha_4 & \geq \alpha_1 \nonumber\\
      g_1,g_3,g_4 & : &  \alpha_1+\alpha_2+\alpha_3+2\alpha_4 & \geq 2\alpha_1+\alpha_3 \nonumber\\
      g_1,g_3,g_4,g_2 & : & \alpha_1+\alpha_2+\alpha_3+2\alpha_4 & \geq 2\alpha_1 + \alpha_2 + 2\alpha_3 + \alpha_4  \nonumber%
  \end{align}
      \item \SDEF{} from agent 2's perspective towards agent 1
      \begin{align}
      g_1 & : & \alpha_3+\alpha_4 & \geq \alpha_1+\alpha_2\label{eqn:2-for-1-top_one}\\
      g_1,g_2 & : & \alpha_1+\alpha_3+\alpha_4 & \geq \alpha_1 + \alpha_2 \nonumber \\
      g_1,g_2,g_4 & : & \alpha_1+\alpha_2+\alpha_3+\alpha_4 & \geq \alpha_1+\alpha_2+\alpha_4 \nonumber \\
      g_1,g_2,g_4,g_3 & : & \alpha_1+2\alpha_2+\alpha_3+\alpha_4 & \geq \alpha_1+\alpha_2+\alpha_3+2\alpha_4 \label{eqn:2-for-1-top_four}
  \end{align}
      \item \SDEF{} from agent 2's perspective towards agent 3
      \begin{align}
      g_1 & : & \alpha_3+\alpha_4 & \geq 0 \nonumber \\
      g_1,g_2 & : & \alpha_1 + \alpha_3+\alpha_4 & \geq \alpha_2 + \alpha_3+\alpha_4 \label{eqn:2-for-3-top_two}\\
      g_1,g_2,g_4 & : & \alpha_1 + \alpha_2 + \alpha_3+\alpha_4 & \geq \alpha_1 + \alpha_2 + 2\alpha_3+\alpha_4\label{eqn:2-for-3-top_three}\\
      g_1,g_2,g_4,g_3 & : & \alpha_1 + 2\alpha_2 + \alpha_3 + \alpha_4 & \geq 2\alpha_1+\alpha_2+2\alpha_3+\alpha_4\label{eqn:2-for-3-top_four} 
  \end{align}
      \item \SDEF{} from agent 3's perspective towards agent 1
      \begin{align}
      g_2 & : & \alpha_2+\alpha_3+\alpha_4 & \geq 0 \nonumber\\
      g_2,g_3 & : & \alpha_1+\alpha_2+\alpha_3+\alpha_4 & \geq \alpha_3+\alpha_4 \nonumber \\
      g_2,g_3,g_4 & : & 2\alpha_1+\alpha_2+2\alpha_3+\alpha_4 & \geq \alpha_3+2\alpha_4 \nonumber \\
      g_2,g_3,g_4,g_1 & : & 2\alpha_1 + \alpha_2 + 2\alpha_3 + \alpha_4 & \geq \alpha_1 + \alpha_2 + \alpha_3 + 2\alpha_4 \nonumber %
  \end{align}
      \item \SDEF{} from agent 3's perspective towards agent 2
      \begin{align}
      g_2 & : & \alpha_2+\alpha_3+\alpha_4 & \geq \alpha_1 \nonumber \\
      g_2,g_3 & : & \alpha_1+\alpha_2+\alpha_3+\alpha_4 & \geq \alpha_1+\alpha_2 \nonumber \\
      g_2,g_3,g_4 & : & 2\alpha_1+\alpha_2+2\alpha_3+\alpha_4 & \geq \alpha_1 + 2\alpha_2 \nonumber \\
      g_2,g_3,g_4,g_1 & : & 2\alpha_1+\alpha_2+2\alpha_3+\alpha_4 & \geq \alpha_1 + 2\alpha_2 + \alpha_3+\alpha_4 \nonumber%
  \end{align}
  \end{itemize}

Since $\alpha_i \geq 0$ for all $i \in [4]$, \Cref{eqn:2-for-3-top_three} implies that
\begin{equation}
    \alpha_3 = 0.
    \label{eqn:temp_1}
\end{equation}

Substituting this value of $\alpha_3$ in \Cref{eqn:1-for-2-top_one,eqn:2-for-1-top_one} gives
\begin{equation}
    \alpha_1+\alpha_2 = \alpha_4.
    \label{eqn:temp_2}
\end{equation}

Further, substituting the value of $\alpha_3$ from \Cref{eqn:temp_1} into \Cref{eqn:2-for-3-top_two,eqn:2-for-3-top_four}  gives
\begin{equation*}
    \alpha_1=\alpha_2.
\end{equation*}

Moreover, \Cref{eqn:2-for-1-top_four,eqn:1-for-2-top_four} give
\begin{equation*}
    \alpha_2 = \alpha_4.
\end{equation*}

Thus we have $\alpha_1 = \alpha_2 = \alpha_4$. This, along with \Cref{eqn:temp_2}, implies that 
$\alpha_1 = \alpha_2 = \alpha_4 = 0$.
However, this contradicts the requirement that $\sum_{i=1}^{4} \alpha_i = 1$. Thus, for the above instance, ex-ante sd-envy-freeness is incompatible with ex-post \EFX{}.
\end{proof}

\section{Uniform Random Permutation}\label{apdx:lexicographic}

\UniformPermutation*

\begin{proof} The ex-post $\EFX{}+\PO{}$ guarantee follows from \Cref{prop:EFX+PO_characterization}, as previously discussed. So, we will only focus on proving the ex-ante $\frac{1}{2}$-\EF{} implication.

Fix any pair of agents $i$ and $j$. Pick an arbitrary permutation $\sigma$ of the agents in $N$. Let $k_i$ and $k_j$ denote the positions of the agents $i$ and $j$, respectively, in $\sigma$. Assume, without loss of generality, that $k_i < k_j$, i.e., agent $i$ goes before agent $j$ in $\sigma$.

Define the \emph{partner} of the permutation $\sigma$, denoted by $\tau$, as the permutation obtained by swapping the positions of agents $i$ and $j$ in $\sigma$ while keeping the positions of all other agents unchanged. 

Consider the execution of the \UnifPerm{} algorithm when the sampled permutation is $\sigma$. Let agent $i$ pick the good $g^{\sigma}_i$ and let agent $j$'s bundle be $A^{\sigma}_j$. Note that if $k_j = n$ (i.e., when $j$ is the last agent in $\sigma$), the bundle $A^{\sigma}_{j}$ could contain multiple goods. Similarly, when the sampled permutation is $\tau$, let agent $j$ pick the good $g^{\tau}_j$ and let agent $i$'s bundle be denoted by $A^{\tau}_{i}$. 

We will now analyze agent $i$'s envy towards agent $j$ restricted to the permutations $\sigma$ and $\tau$. 

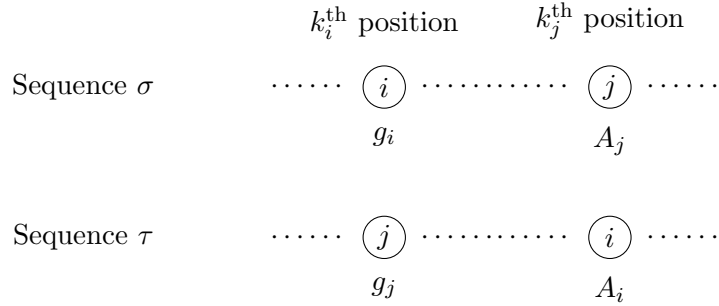
\begin{figure}[h]
    \centering
    \begin{tikzpicture}
        \centering
        \node at (-4,-2) {Sequence $\sigma$};
        \node at (-4,-4) {Sequence $\tau$};
        \node at (-1, -2) {$\cdots \cdots$};
        \node at (-1, -4) {$\cdots \cdots$};
        \node[draw, circle, minimum size=16pt, inner sep=0pt] (g1) at (0, -2) {$i$};
        \node[below] at (0,-2.4) {$g^{\sigma}_i$};
        \node[above] at (0,-1.5) {$k_i^\textup{th}$ position};
        \node[draw, circle, minimum size=16pt, inner sep=0pt] (g2) at (0, -4) {$j$};
        \node[below] at (0,-4.4) {$g^{\tau}_j$};
        \node at (1, -2) {$\cdots \cdots$};
        \node at (1, -4) {$\cdots \cdots$};
        \node at (2, -2) {$\cdots \cdots$};
        \node at (2, -4) {$\cdots \cdots$};
        \node[draw, circle, minimum size=16pt, inner sep=0pt] (g3) at (3, -2) {$j$};
        \node[below] at (3,-2.4) {$A^{\sigma}_j$};
        \node[draw, circle, minimum size=16pt, inner sep=0pt] (g4) at (3, -4) {$i$};
        \node[above] at (3,-1.5) {$k_j^\textup{th}$ position};
        \node[below] at (3,-4.4) {$A^{\tau}_i$};
        \node at (4, -2) {$\cdots \cdots$};
        \node at (4, -4) {$\cdots \cdots$};
    \end{tikzpicture}
    \caption{We consider permutations $\sigma$ and its partner $\tau$ where the positions of agents $i$ and $j$ are swapped, while all other agents remain unchanged. Let $k_i$ and $k_j$ denote the positions of agents $i$ and $j$ in $\sigma$, respectively. If $k_j = n$, then $A_j$ and $A_i$ could contain multiple goods. }
    \label{fig:randomperm}
\end{figure}

Note that $v_i(g^{\sigma}_i) \ge v_i(A^{\sigma}_j)$ due to the lexicographicity of the valuation functions: $i$ prefers $g^{\sigma}_i$ to every good in $A^{\sigma}_j$ and therefore, $i$ prefers $g^{\sigma}_i$ to the bundle $A^{\sigma}_j$. In both $\sigma$ and $\tau$, the agents occurring before the position $k_{i}$ are ordered identically. Therefore, in both sequences, the exact same goods are picked until position $k_i$. Hence, the good $g^{\tau}_j$ is available at agent $i$'s turn in $\sigma$. Therefore, $v_i(g^{\sigma}_i) \ge v_i(g^{\tau}_j)$.

The distribution generated by sampling a permutation $\sigma$ uniformly over all permutations where $i$ appears before $j$, and then swapping the positions of $i$ and $j$ with $1/2$ probability is equivalent to uniform distribution over all permutations. Let $\X$ be the output of the \UnifPerm{} algorithm. Then, 
\begin{align*}
    \E[v_i(\X_j)] = \E_{\sigma}\left [\frac12 v_i(g^{\tau}_j) + \frac12 v_i(A^{\sigma}_j) \right ] \le \E_{\sigma}\left [v_i(g^{\sigma}_i) \right ] \le 2\E_{\sigma}\left [\frac12 v_i(g^{\sigma}_i) + \frac12 v_i(A^{\tau}_i) \right ] = 2\E[v_i(\X_i)]. 
\end{align*}
Here, $\E_\sigma$ denotes the expectation where $\sigma$ is sampled uniformly from the set of permutations where $i$ appears before $j$, and $\tau$ denotes the partner permutation of $\sigma$.
\end{proof}

\section{Correctness of Algorithm~\ref{alg:utse}}
\label{sec:6-over-7-EF}
In this section, we prove the following result.

\UtseUpperBound*

We start with the ex-post guarantee.

\begin{restatable}[]{proposition}{ex-post guarantees} Algorithm \ref{alg:utse} returns a randomized allocation that is ex-post \EFX{}\textnormal{+}\PO{}.
\label{prop:ex-post utse}
\end{restatable}
\begin{proof}
Let the output of Algorithm \ref{alg:utse} be the randomized allocation $\Z$, and let $X$ be the fractional allocation generated by running the eating algorithm for one unit of time.

A picking sequence is a sequence $(a_1, a_2, \dots,a_m)$ of length $m$, where $a_i \in N$ for every $i \in [m]$. A picking sequence naturally induces an allocation as follows: agents take turns following the sequence, with each picking its favorite remaining item. For example, given an instance with four goods, the sequence $\sigma = (1, 2, 3, 2)$ corresponds to an allocation where agents $1$, $2$, and $3$ pick their favorite remaining good in that order, and finally agent~$2$ picks the leftover good. It can be observed that under lexicographic preferences, an allocation is \PO{} if and only if it is induced by some picking sequence~\cite{ABL+19efficient,HSV+21fair}. In \Cref{subsec:Proof_sdEF_EF1+PO_lexicographic}, we show that all integral allocations in the Birkhoff-von Neumann decomposition of $X$ can be induced by picking sequences. This picking sequence can be extended to include all the goods in the {\em tail} which all go to one agent. Therefore, any integral allocation $Z$ in the support of $\Z$ is guaranteed to be \PO{}.

Let $A$ be any integral allocation given by the B-vN decomposition of $X$. Each agent receives exactly one good in $A$. The remaining $m-n$ unallocated goods (in $A$) are assigned to an agent that receives a \textit{last consumed good} in $A$. Observe that such an agent is unenvied in $A$ and due to the lexicographic property, it remains unenvied in the final allocation at the end of Phase 2. Since all the other agents receive exactly one good, the final allocation is guaranteed to be EFX.
\end{proof}

Towards proving the ex-ante guarantee in \Cref{thm:utse-upper-bounnd}, we define the {\em last consumed mass}, denoted by $k$, as the total consumed amount of the goods in $L$, i.e., 
$$k \coloneqq \sum_{i \in N} \sum_{g \in L} X_{i, g}.$$

As we show in \Cref{lem:nl-size}, $k$ is guaranteed to be a positive integer. We show that our algorithm computes a randomized allocation that is ex-ante $\frac{3k}{3k+1}$-\EF{} when $k > 1$ and ex-ante \EF{} when $k = 1$. The basic approach for this proof is to fix two agents $i$ and $j$, and bound the values $\E[v_i(\Z_j)] - \E[v_i(\Z_i)]$ and $\E[v_i(\Z_i)]$. By showing an upper bound on the ratio of these two values, i.e., $\frac{\E[v_i(\Z_j)] - \E[v_i(\Z_i)]}{\E[v_i(\Z_i)]} \leq \frac{1}{3k}$, we obtain the desired approximation of ex-ante envy-freeness.

We require two properties of the eating algorithm. For any $z \in [0, 1]$, let $X(z)$ denote the fractional allocation that results from running the eating algorithm for $z$ units of time. 

\begin{itemize}
    \item Property (a): Each agent is allocated a positive fraction of {\em exactly one} good in $L$ under the fractional allocation $X$. If any good is in $L$, it remains available to be eaten till the very end of the eating algorithm. Therefore, if an agent starts consuming some good in $L$, it continues to do so till the end of the eating process.
    \item Property (b): The eating algorithm satisfies \emph{anytime envy-freeness}. Formally, for any $i, j \in N$ and $z \in [0, 1]$, $v_i(X_i(z)) \ge v_i(X_j(z))$, where $X_i \coloneqq (X_{i,1},\dots,X_{i,m})$ and $X_j \coloneqq (X_{j,1},\dots,X_{j,m})$ denote the fractional bundles of agents $i$ and $j$ under $X$, respectively.
\end{itemize}

We start by proving an important result about the B-vN decomposition. For this, we define the set $U$ as the set of \emph{uneaten} goods in the fractional allocation $X$; that is, $U \coloneqq \{g \in M: X_{i, g} = 0 \text{ for all } i \in N\}$.

\begin{lemma}
Fix some partial allocation $Z \in \{A^1,\dots,A^q\}$ in the support of the B-vN decomposition of $X$. The following two properties hold:
\begin{enumerate}
    \item[(i)] The set of unallocated goods in $Z$ is a subset of $L \cup U$, and 
    \item[(ii)] $|N_{L}| = k \ge 1$ where $N_L = \{i \in N: Z_i \cap L \ne \emptyset\}$, and
\end{enumerate}
Specifically, $k = |N_L| \ge 1$ implies that $k$ is a positive integer.
\label{lem:nl-size}
\end{lemma}
\begin{proof}
Consider any good $g$ that is positively consumed in $X$, i.e., $X_{i, g} > 0$ for some $i \in N$, but $g$ is not included in the set $L$. We claim that $\sum_{h \in N} X_{h, g} = 1$ for such a good $g$; in other words, the good $g$ must be fully consumed. Indeed, if $g \notin L$, then agent $i$ must have started consuming the good $g$ but switched to another good before the eating algorithm terminates. Since agents only switch to another good in the eating algorithm if the previous good is fully consumed, we must have that $\sum_{h \in N} X_{h, g} = 1$. Let $M'$ be the set of positively consumed goods in $X$ that are not in $L$; note that $M' \cup L \cup U = M$.

Since each good in $M'$ is fully consumed, it is allocated in every allocation in the support of the B-vN decomposition of $X$. This proves the first property. 

The eating algorithm runs for one unit of time and agents eat at the speed of one good per unit time; therefore, the total consumed amount across all goods is $\sum_{i \in N} \sum_{g \in M} X_{i, g}$ is $|N|=n$. Therefore the total mass of items in $M'$ which are consumed is
\begin{align*}
    \sum_{i \in N} \sum_{g \in M'} X_{i, g} = \sum_{i \in N} \sum_{g \in M} X_{i, g} - \sum_{i \in N} \sum_{g \in L} X_{i, g} = n - k.
\end{align*}
As discussed above, each item in $M'$ is fully consumed. Therefore, the left hand side of the above inequality is an integer, which implies $k$ is an integer as well. Additionally, this means that there are exactly $n-k$ goods in $M'$.

Recall that the partial allocation $Z$ assigns exactly one good to each agent from the set of goods that are consumed in the eating algorithm. Out of the $n$ agents, exactly $n-k$ agents receive goods in $M'$. This leaves $k$ agents, who are allocated goods in $L$. $k$ is trivially positive, and from the above discussion, $k$ is an integer. We can therefore conclude that $k$ is at least $1$. This proves the second property.
\end{proof}

The above lemma shows that while the \emph{set} $N_L$ may vary depending on the sampled partial allocation $Z$, the number of agents in the set $N_L$ is the same for any $Z$. Our next observation bounds from below agent $i$'s expected value for its own bundle in the randomized allocation $\Z$ output by Algorithm \ref{alg:utse}, namely $\E[v_i(\Z_i)]$. 

\begin{lemma}
Suppose agent $i$ consumes $x \ge 0$ amount of its last consumed good $g_i$ in the fractional allocation $X$. Then, $\E[v_i(\Z_i)] \ge (2-x) \cdot v_i(L \cup U \setminus \{g_i\})$.
\label{obs:lower-bound}
\end{lemma}
\begin{proof}
By definition of the last consumed good, agent $i$ spends $(1-x)$ units of time consuming goods it prefers over $g_i$. Among all the goods consumed by agent $i$ during the eating algorithm that are better than $g_i$, i.e., the goods in $\{j \in M: X_{i,j}>0 \text{ and } v_i(j) > v_i(g_i)\}$, let $g^*$ be its least-valued good. Since agent $i$ has lexicographic valuations and since agent $i$ consumes $g^*$ before consuming any good in $L \cup U$, we have that $v_i(g^*) \ge v_i(L \cup U)$. Additionally, since $g_i$ is agent $i$'s most preferred good in $L \cup U$, we have $v_i(g_i) \ge v_i(L \cup U \setminus \{g_i\})$. Therefore,
\begin{alignat*}{3}
    \E[v_i(\Z_i)] 
    & \ge v_i(X_i) & \text{(since $\Z_i$ is a completion of $X_i$)}\\
    & \ge (1-x) v_i(g^*) + xv_i(g_i) & \text{(by definition of $g^*$)}\\
    & \ge (1-x) v_i(L \cup U) + xv_i(g_i) & \text{(due to lexicographic valuations)}\\
    & = (1-x) v_i(L \cup U \setminus \{g_i\}) + (1-x) v_i(g_i) + xv_i(g_i) & \text{(due to additive valuations)}\\
    & = (1-x) v_i(L \cup U \setminus \{g_i\}) + v_i(g_i) & \\
    & \ge (2-x) v_i(L \cup U \setminus \{g_i\}) & \text{(since $v_i(g_i) \ge v_i(L \cup U \setminus \{g_i\})$)},
\end{alignat*}
as desired.
\end{proof}

We will now show an upper bound on agent $i$'s expected value for agent $j$'s bundle, namely $\E[v_i(\Z_j)]$. Towards this goal, we first establish an upper bound on the probability that an agent receives a good in $L \cup U$ other than its most preferred one. 

We write $q(g)$ to denote the amount of good $g$ that is allocated (or eaten) in the fractional allocation $X$, i.e., $q(g) \coloneqq \sum_{h \in N} X_{h,g}$. Note that $q(g) \in [0, 1]$.

\begin{lemma}
Suppose agent $i$ consumes $x \geq 0$ amount of its last consumed good $g_i$ in the fractional allocation $X$. For any good $g \in L \cup U \setminus \{g_i\}$, the probability that agent $i$ receives good $g$ in the randomized allocation $\Z$ is at most $\frac1k \min\{x, 1 - q(g)\}$.
\label{obs:fractional}
\end{lemma}
\begin{proof}
By Property (a), agent $i$ is not allocated any fraction of the goods in $L \cup U \setminus \{g_i\}$ in $X$
Therefore the only way for $i$ to receive these goods is as part of the tail. 
Let $Z \sim \Z$ denote an allocation sampled from the randomized allocation $\Z$ output by Algorithm \ref{alg:utse}. This allocation $Z$ contains a part $A^Z$ which is the partial allocation obtained from the B-vN decomposition, and a tail $T^Z$.
Using \Cref{lem:nl-size}, the probability that agent $i$ is assigned the tail and the tail contains good $g$ is 
\begin{align*}
    \frac1k \Pr_{Z \sim \Z}[A^Z_i = \{g_i\} \text { and } g \in T^Z] & = \frac1k \Pr_{Z \sim \Z}[A^Z_i = \{g_i\} \text { and } g \text{ is unallocated in } A^Z]\\
    & \leq \frac1k \min \left\{\Pr_{Z \sim \Z} \left[A^Z_i = \{g_i\} \right], \Pr_{Z \sim \Z}\left [g \text{ is unallocated in } A^Z \right ]\right\}.
\end{align*}
From the construction of the randomized allocation $\Z$, the probability that $A^Z = A^{\ell}$ for some $A^{\ell}$ in the B-vN decomposition of $X$ is equal to $p_\ell$: the probability (or weight) of the allocation $A^{\ell}$ in the B-vN decomposition.
Therefore, the probability that $i$ gets $g_i$ in $A^Z$ is $x$, and the probability that $g$ is unallocated in $A^Z$ is $1-q(g)$. 
\end{proof}

The next lemma shows an upper bound on the expected envy of agent $i$ towards agent $j$, namely $\E[v_i(\Z_j)] - \E[v_i(\Z_i)]$. In fact, we prove something stronger: we show that $\E[v_i(\Z_j)] - v_i(X_i)$ is small when $k$ is large. 

\begin{lemma}
Suppose agent $i$ consumes $x \geq 0$ amount of its last consumed good $g_i$ in the fractional allocation $X$. Then, for any other agent $j$, $\E[v_i(\Z_j)] - v_i(X_i) \le \frac1k \min\{x, 1-x\} v_i(L \cup U \setminus \{g_i\})$.
\label{lem:ij-fractional}
\end{lemma}
\begin{proof}
Suppose agent $j$ consumes $y \geq 0$ amount of its last consumed good $g_j$, i.e., $X_{j,g_j} = y$. 
Our argument will consider two cases: when $g_i$ and $g_j$ are distinct and when the $g_i$ and $g_j$ are identical.

\begin{figure}[h]
    \centering
        \begin{tikzpicture}[scale=0.8]
            \centering
            \draw[blue,thick] (0,0) -- node[below] {$1-x$} (3,0);
            \draw[blue,thick,dashed] (3,0) -- node[below] {$x$} (5,0);
            \draw[red,thick] (0,-1) -- node[below] {$1-y$} (4,-1);
            \draw[red,thick,dashed] (4,-1) -- node[below] {$y$} (5,-1);
            \node at (-0.75, 0) (mynode) {agent $i$};
            \node at (-0.75, -1) (mynode) {agent $j$};
            \draw[blue,thick] (9,0) -- node[below] {$1-x$} (13,0);
            \draw[blue,thick,dashed] (13,0) -- node[below] {$x$} (14,0);
            \draw[red,thick] (9,-1) -- node[below] {$1-y$} (12,-1);
            \draw[red,thick,dashed] (12,-1) -- node[below] {$y$} (14,-1);
            \node at (8.25, 0) (mynode) {agent $i$};
            \node at (8.25, -1) (mynode) {agent $j$};
        \end{tikzpicture}
    \caption{Eating timelines of agents $i$ and $j$ when $x > y$ (left) and $y > x$ (right). The dashed lines denote the time an agent spends on its last consumed good.}
    \label{fig:eating-timeline}
\end{figure}

\textbf{Case (a): $g_j \ne g_i$}

We first establish an upper bound on $\E[v_i(\Z_j)]$. The randomized bundle $\Z_j$ of agent $j$ consists of the fractional goods assigned to it under $X$, namely $X_j$, as well as the tail goods.

To show an upper bound on $v_i(X_j)$, we will break down the eating timeline of agent $j$ as follows (also see \Cref{fig:eating-timeline}): 
\begin{itemize}
    \item For the time interval $[0,1-\max\{x,y\}]$, we consider the fractional bundle $X_j(1-\max\{x,y\})$ consumed by agent $j$.
    \item Next, if $x > y$, we consider the interval $[1-x,1-y]$ during which agent $j$ consumes some set of goods $M'$ disjoint from $L$. Among the goods in $M'$, let $g^*$ be agent $i$'s most preferred good. Then, agent $i$'s value for the fractional goods eaten by agent $j$ during the interval $[1 - max\{x,y\},1-y]$ can be bounded from above as $\max\{x-y, 0\} v_i(g^*)$. %
    \item For the interval $[1-y,1]$, agent $j$ consumes the good $g_j$. Thus, agent $i$'s value for the fractional goods eaten by agent $j$ during the interval $[1-y,1]$ is equal to $y v_i(g_j)$.
\end{itemize}

Together, these observations imply that
$$v_i(X_j) \le v_i(X_j(1-\max\{x,y\})) + \max\{x-y, 0\} v_i(g^*) + yv_i(g_j).$$

Let us now consider the assignment of tail goods to agent $j$. We have already taken into account the good $g_j$ in the above calculation; thus, we only need to consider the goods in $L \cup U \setminus \{g_j\}$. From \Cref{obs:fractional}, we know that for any good $g \in L \cup U \setminus \{g_j\}$, the probability that agent $j$ receives good $g$ in the randomized allocation $\Z$ is at most $\frac1k \min\{y, 1 - q(g)\}$. Thus, agent $i$'s value for the tail goods (other than $g_j$) assigned to agent $j$ is at most $\sum_{g \in L \cup U \setminus \{g_j\}} \frac1k \min\{y, 1- q(g)\} v_i(g)$.

Combining the above observations, we can write the following upper bound on $\E[v_i(\Z_j)]$:
\begin{align}
    \E[v_i(\Z_j)] & \le v_i(X_j(1-\max\{x,y\})) + \max\{x-y, 0\} v_i(g^*) + yv_i(g_j) \nonumber \\
    & \pushright{ + \sum_{g \in L \cup U \setminus \{g_j\}} \frac1k \min\{y, 1- q(g)\} v_i(g)} \nonumber \\
    &\le v_i(X_j(1-\max\{x,y\})) + \max\{x-y, 0\} v_i(g^*) \nonumber \\
    & \pushright{+ \frac1k \min\{y, 1 - x\}v_i(g_i) + yv_i(L \cup U \setminus \{g_i\})}. \label{eqn:agent_j_exp_val}
\end{align}

The second inequality follows from the observation that $\frac1k \min\{y, 1- q(g)\}$ is at most $\frac1k \min\{y, 1- x\}$ when $g = g_i$ and $y$ when $g \in L \cup U \setminus \{g_i,g_j\}$.

We next establish a lower bound on $v_i(X_i)$. To show a lower bound on $v_i(X_i)$, we will break down the eating timeline of agent $i$ as follows (also see \Cref{fig:eating-timeline}): 
\begin{itemize}
    \item For the time interval $[0,1-\max\{x,y\}]$, we consider the fractional bundle $X_i(1-\max\{x,y\})$ consumed by agent $i$.
    \item Next, if $x < y$, we consider the interval $[1-y,1-x]$ during which agent $i$ consumes goods that it prefers over $g_i$. By lexicographic valuations, each of these goods is preferred over the set $L \cup U$.
    \item For the interval $[1-x,1]$, agent $i$ consumes the good $g_i$.
\end{itemize}

Together, the above observations imply that
\begin{align}
    v_i(X_i) &\ge v_i(X_i(1-\max\{x,y\})) + \max\{y-x, 0\} v_i(L \cup U) + xv_i(g_i) \nonumber \\
    &= v_i(X_i(1-\max\{x,y\})) + \max\{y-x, 0\} v_i(L \cup U\setminus \{g_i\}) + \max\{x, y\}v_i(g_i).\label{eqn:agent_i_exp_val}
\end{align}

Taking the difference of \Cref{eqn:agent_j_exp_val,eqn:agent_i_exp_val} and using $v_i(X_j(1-\max\{x,y\})) \le v_i(X_i(1-\max\{x,y\}))$, we get
\begin{align*}
    & \E[v_i(\Z_j)] - v_i(X_i) \le \\
    & \max\{x-y, 0\} v_i(g^*) + \left [\frac1k \min\{y, 1 - x\} - \max\{x, y\} \right ]v_i(g_i) + \min\{x, y\}v_i(L \cup U \setminus \{g_i\}).
\end{align*}

The coefficient of $v_i(g_i)$ in the above expression is nonpositive, so we can apply $v_i(g_i) \geq v_i(L \cup U \cup \{g^*\} \setminus \{g_i\})$. This is because $v_i(g_i) \geq v_i(L \cup U \setminus \{g_i\})$, $v_i(g_i) \geq v_i(g^*)$, and the agents have lexicographic valuations. (Recall that $g^*$ is agent $i$'s favorite good among those that are consumed by agent $j$ during the interval $[1-x,1-y]$. If $x < y$, then $\{g^*\} = \emptyset$.)

The aforementioned substitution gives
\begin{align*}
    \E[v_i(\Z_j)] - v_i(X_i) &\le \left [\max\{x-y, 0\} + \frac1k \min\{y, 1-x\} - \max\{x, y\} \right ] v_i(g^*) \\ &\qquad + \left [\frac1k \min\{y, 1 - x\} - \max\{x, y\} + \min\{x, y\} \right ]v_i(L \cup U \setminus \{g_i\}) \\
    &\leq 0 + \left[ \frac{1}{k} \min(x, 1-x) + \frac{\max\{0, y-x\}}{k} - \mid y - x \mid \right] v_i(L \cup U \setminus \{g_i\})\\
    &\le \frac1k \min\{x, 1 - x\} v_i(L \cup U \setminus \{g_i\}). 
\end{align*}

The second inequality follows since the coefficient of $v_i(g^*)$ is nonpositive. This establishes the desired inequality for the case where $g_i$ and $g_j$ are distinct.

\smallskip

\textbf{Case (b): $g_j = g_i$}

Following the same arguments as in the previous case, we can bound $\E[v_i(\Z_j)]$ from above as follows:
\begin{align*}
    \E[v_i(\Z_j)] & \le v_i(X_j(1-\max\{x,y\})) + \max\{x-y, 0\} v_i(g^*) + yv_i(g_i) \\
    & \pushright{+ \sum_{g \in L \cup U \setminus \{g_i\}} \frac1k \min\{y, 1- q(g)\}v_i(g)} \\
    & \le v_i(X_j(1-\max\{x,y\})) + \max\{x-y, 0\} v_i(g^*)  + yv_i(g_i) + \frac{y}{k}v_i(L \cup U \setminus \{g_i\}). 
\end{align*}

Note the subtle difference in the second inequality; we bound the value $\frac1k \min\{y, 1-q(g)\}$ from above as $\frac{y}{k}$ for all $g \in L \cup U \setminus \{g_i\}$. We can similarly bound $v_i(X_i)$ from below using \Cref{eqn:agent_i_exp_val}:
\begin{align*}
    v_i(X_i) &\ge v_i(X_i(1-\max\{x,y\})) + \max\{y-x, 0\} v_i(L \cup U \setminus \{g_i\}) + \max\{x, y\}v_i(g_i).
\end{align*}

Taking the difference, we get
\begin{align*}
    & \E[v_i(\Z_j)] - v_i(X_i) \\
    &\le \max\{x-y, 0\} v_i(g^*) - \max\{x-y, 0\}v_i(g_i) + \left (\frac{y}{k} - \max\{y- x, 0\} \right )v_i(L \cup U \setminus \{g_i\}) \\
    &\le \max\{x-y, 0\} v_i(g^*) - \max\{x-y, 0\}v_i(g_i) + \frac1k \min\{x, y\}v_i(L \cup U \setminus \{g_i\}) \\
    &\le \frac1k \min\{x, y\}v_i(L \cup U \setminus \{g_i\}) \\
    &\le \frac1k \min\{x, 1-x\}v_i(L \cup U \setminus \{g_i\}) 
\end{align*}

The second inequality uses $\max\{y-x, 0\} \geq \frac1k\max\{y-x, 0\}$. The third inequality follows from the fact that $v_i(g^*) \le v_i(g_i)$. The fourth inequality follows from the observation that both $x$ and $y$ refer to the amount of consumption of the same good $g_i$. Therefore, $x + y \le 1$.
\end{proof}

Combining \Cref{obs:lower-bound,lem:ij-fractional} we get the following result. 

\UtseUpperBound*
\begin{proof}
We first show the ex-ante \EF{} guarantee. Consider any two agents $i$ and $j$, we need to show $\frac{\E[v_i(\Z_i)]}{\E[v_i(\Z_j)]} \ge \frac{3k}{3k+1}$. 

Dividing the bounds obtained in \Cref{obs:lower-bound} and \Cref{lem:ij-fractional}, we get
\begin{align*}
    \frac{\E[v_i(\Z_j)] - \E[v_i(\Z_i)]}{\E[v_i(\Z_i)]} \le \frac{\E[v_i(\Z_j)] - v_i(X_i)}{\E[v_i(\Z_i)]} \le \max_{x \in [0, 1]} \frac1k \frac{\min\{x, 1-x\}}{2-x} \le \frac{1}{3k}.
\end{align*}

The first inequality follows from the fact that $\E[v_i(\Z_i)] \ge v_i(X_i)$. Re-arranging the terms from the above inequality, we prove the ex-ante \EF{} guarantee.
\end{proof}

Since $k \ge 1$, the above theorem proves an ex-ante $\frac34$-\EF{} guarantee. If $k \ge 2$, then we achieve the promised ex-ante $\frac{6}{7}$-\EF{} guarantee. Since $k$ is a positive integer (\Cref{lem:nl-size}), the only case left to improve the guarantee for is when $k = 1$. Our next result shows that when $k =1$, Algorithm \ref{alg:utse} is guaranteed to output an allocation that is ex-ante \EF{}. 

\begin{theorem}\label{thm:low-k}
If $k = 1$, then Algorithm~\ref{alg:utse} outputs a randomized allocation that is ex-ante \EF{}.
\end{theorem}
\begin{proof}
The proof follows from the fact that if $k = 1$, then $|N_L| = 1$ (\Cref{lem:nl-size}). Therefore, the allocation of the tail is deterministic. Additionally, since $|N_L| = 1$, if some agent $i$ receives its last consumed good $g_i$, no other agent receives a good in $L$. This means the tail for the agent $i$ is the set of goods $L \cup U \setminus \{g_i\}$, and $i$ receives this tail in all allocations where it receives $g_i$. 

Assume some agent $i$ consumes $x$ amount of its last consumed good $g_i$. This implies it receives the tail with probability $x$. Therefore,
\begin{align*}
    \E[v_i(\Z_i)] = v_i(X_i) + xv_i(L \cup U \setminus\{g_i\}).
\end{align*}

For any other agent $j$, 
\begin{align*}
    \E[v_i(\Z_j)] - \E[v_i(\Z_i)] &= \E[v_i(\Z_j)] - v_i(X_i) - xv_i(L \cup U \setminus\{g_i\}) \\
    &\le \frac1k \min\{x, 1-x\} v_i(L \cup U \setminus \{g_i\}) - xv_i(L \cup U \setminus\{g_i\}) \\
    &\le 0.
\end{align*}
The first inequality follows from \Cref{lem:ij-fractional}. We can therefore conclude that agent $i$ does not have ex-ante envy towards agent $j$. Since these two agents were picked arbitrarily, we conclude that $\Z$ is ex-ante \EF.
\end{proof}

In the following example, we show that our analysis for Algorithm \ref{alg:utse} is tight. 

\begin{example}[Algorithm \ref{alg:utse} fails ex-ante $\nicefrac{6}{7} + \epsilon$-\EF{}] Consider an instance with four agents and five goods. For $\epsilon>0$ sufficiently small, the valuations are given as follows:  

\begin{table}[h]
    \centering
    \begin{tabular}{c|cccccc}
        & $g_1$ & $g_2$ & $g_3$ & $g_4$ & $g_5$\\  
        \cmidrule{1-6}
        $1$ & $4+\epsilon$ & $2+\epsilon/2$ & $2$ & $\epsilon/4$ & $\epsilon/8$\\  
        $2$ & $4+\epsilon$ & $2+\epsilon/2$ & $2$ & $\epsilon/4$ & $\epsilon/8$\\
        $3$ & $\epsilon/8$ & $\epsilon/4$ & $2+\epsilon/2$ & $4+\epsilon$ & $2$ \\
        $4$ & $\epsilon/8$ & $\epsilon/4$ & $2$ & $4+\epsilon$ & $2+\epsilon/2$ \\
         
    \end{tabular}
    \label{tab: 6/7 ex-ante EF example}
\end{table}

Let $X$ denote the fractional allocation obtained by running the eating algorithm for one unit of time. $X$ can be decomposed into integral allocations $A^1, A^2$ as follows: 

$$\NiceMatrixOptions{code-for-first-row=\scriptstyle,code-for-first-col=\scriptstyle}
\begin{bNiceMatrix}
1/2 & 1/2 & 0 & 0 & 0 \\
1/2 & 1/2 & 0 & 0 & 0 \\ 
0 & 0 & 1/2 & 1/2 & 0 \\
0 & 0 & 0 & 1/2 & 1/2
\end{bNiceMatrix}
=
\frac{1}{2} \cdot
\begin{bNiceMatrix}
1 & 0 & 0 & 0 & 0 \\
0 & 1 & 0 & 0 & 0 \\ 
0 & 0 & 0 & 1 & 0 \\
0 & 0 & 0 & 0 & 1
\end{bNiceMatrix} 
+
\frac{1}{2} \cdot
\begin{bNiceMatrix}
0 & 1 & 0 & 0 & 0 \\
1 & 0 & 0 & 0 & 0 \\ 
0 & 0 & 1 & 0 & 0 \\
0 & 0 & 0 & 1 & 0
\end{bNiceMatrix} 
$$

This decomposition satisfies all the required properties of the B-vN decomposition. In $A^1$, $N_L = \{2, 4\}$ and in $A^2$: $N_L = \{1, 3\}$. Moreover, the set of unallocated goods is different in each allocation: $g_3$ is unallocated in $A_1$ and $g_5$ is unallocated in $A_2$. Assuming we follow the tail assignment procedure of Algorithm \ref{alg:utse}, let us analyze agent $1$'s envy towards agent $2$.

\begin{align*}
    \E[v_1(X_1)] = \frac{1}{2}[v_1(g_1)+v_1(g_2)+\frac{1}{2}v_1(g_5)] = 3 + \frac{25\epsilon}{32} \\ 
    \E[v_1(X_2)] = \frac{1}{2}[v_1(g_1)+v_1(g_2)+\frac{1}{2}v_1(g_3)] = \frac{7}{2} + \frac{3\epsilon}{4}
\end{align*}

Taking the ratio of agent 1's value for its own bundle  and agent 2's bundle under $\Z$:

 \begin{align*}    
 \frac{\E[v_1(\Z_1)]}{\E[v_1(\Z_2)]} &= \frac{3+\frac{25\epsilon}{32}}{\frac{7}{2} + \frac{3\epsilon}{4}} \le \frac67 + \epsilon
 \end{align*}

 \label{ex:utse-tight}
\end{example}

\section{Correctness of Algorithm~\ref{alg:dependent-rounding}}
\label{sec:main-algo}
We follow the same notation as the previous section and introduce the following new notation to analyze the probability that any agent receives some good in $L \cup U$. Let $A_{g}$ be the set of agents whose last consumed good is $g$. Let $E_{ig}$ be the event that some agent $i$ is allocated the super-good together with some agent in $A_g$. We use this notation to compute the probability that some agent $i$ receives an item in $L \cup U$. 

\begin{lemma}\label{lem:depround-probabilities}
Let agent $i$ consume $x$ amount of its last consumed good $g_i$, and let $g$ be some good in $L \cup U$. The probability that $i$ receives $g$ is exactly:
\begin{enumerate}
    \item[(i)] $x - \frac12 \Pr[E_{ig_i}]$ if $g = g_i$, and 
    \item[(ii)] $\frac12(x - \Pr[E_{ig}])$ otherwise.
\end{enumerate}
\end{lemma}
\begin{proof}
Agent $i$ is allocated the super-good with probability $x$. Whenever $i$ is allocated the super-good with an agent from $A_{g_i}$, it receives the good $g_i$ with probability $1/2$ (if it comes first in the permutation). In all other cases, it receives the good $g_i$ with probability $1$. Therefore, 
\begin{align*}
    \Pr[g_i \in Z_i] = (x - \Pr[E_{ig_i}]) + \frac12 \Pr[E_{ig_i}] = x - \frac12 \Pr[E_{ig_i}],
\end{align*}

For any other item $g$, $i$ is allocated $g$ whenever $i$ is allocated the super-good together with an agent who is not in $A_{g}$, and $i$ is picked to be second in the random permutation. The probability of this happening is $\frac12(x - \Pr[E_{ig}])$. 
\end{proof}

Note that we are now using the exact probability of an agent receiving a good in $L \cup U$, as opposed to the weak bounds we used for Algorithm \ref{alg:utse} (\Cref{obs:fractional}). We also present a lemma that bounds these probabilities using negative correlation. 

\begin{lemma}\label{lem:depround-basic}
Fix any set of agents $N' \subseteq N$. For any $i \notin N'$, the probability that $i$ receives the super-good together with an agent in $N'$ is at most $X_{i, g_i} \sum_{j \in N'} X_{j, g_j}$ where $X$ is the fractional allocation constructed by the eating algorithm. Specifically, this implies that $\Pr[E_{ig_i}] \le x(1-x)$.
\end{lemma}
\begin{proof}
\begin{align*}
    \Pr[s \in Z_i \land s \in Z_j, j \in N'] = \sum_{j \in N'} \Pr[s \in Z_i \land s \in Z_j] \le \sum_{j \in N'} \Pr[s \in Z_i]\Pr[s \in Z_j] = X_{i, g_i} \sum_{j \in N'} X_{j, g_j}.
\end{align*}
The first equality follows from the fact that exactly two agents receive the super-good. The inequality follows from the negative correlation property. Note that $\Pr[E_{ig_i}] \le x(1-x)$ since the total amount $g_i$ is consumed in $X$ is at most $1$.
\end{proof}

Similar to \Cref{obs:lower-bound}, we establish a lower bound for agent utilities in the allocation $Z$ output by Algorithm \ref{alg:dependent-rounding}. 

\begin{lemma}\label{lem:depround-lower-bound}
Suppose agent $i$ consumes $x$ amount of its last consumed good $g_i$, and let $\hat{g}$ be a good not in $L \cup U$ that is less preferred by $i$ to $g_i$. Then $\E[v_i(Z_i)] \ge (1-\frac12x(1-x)) v_i(g_i) + (1-x) v_i(L \cup U \cup \{\hat{g}\} \setminus \{g_i\})$.
\end{lemma}
\begin{proof}
This proof follows similarly to \Cref{obs:lower-bound}. Agent $i$ spends $(1-x)$ units of time consuming goods it prefers to $g_i$, and therefore to all goods in $L \cup U \cup \{\hat{g}\}$. Since $i$ has lexicographic valuations, each good consumed in the first $(1-x)$ units of time has value at least $L \cup U \cup \{\hat{g}\}$. Additionally, agent $i$ receives the good $g_i$ with probability $x - \frac12 \Pr[E_{ig_i}]$. Therefore,
\begin{align*}
    \E[v_i(\Z_i)] &\ge (1-x) v_i(L \cup U \cup \{\hat{g}\}) + \left (x - \frac12\Pr[E_{ig_i}] \right )v_i(g_i) \\
    & \ge (1-x) v_i(L \cup U \cup \{\hat{g}\} \setminus \{g_i\}) + \left (1-\frac12\Pr[E_{ig_i}] \right ) v_i(g_i).
\end{align*}
The lemma follows from upper bounding $\Pr[E_{ig_i}]$ using \Cref{lem:depround-basic}.
\end{proof}

Our next lemma provides a sufficient condition to upper bound the envy ratio. We need this condition since in some cases some agent $i$ could be allocated $g_i$ with a lower probability than $j$ is allocated $g_i$. Due to this, the tricks used in \Cref{lem:ij-fractional} no longer suffice. 

\begin{lemma}\label{lem:upperbound-trick}
Assume agent $i$ has lexicographic valuations. Fix some set of goods $M'$. Consider the expressions $\sum_{g \in M'} \alpha_g v_i(g)$ and $\sum_{g \in M'} \beta_g v_i(g)$ where $\alpha_g \in \mathbb{R}$ and $\beta_g \in \mathbb{R}^{\ge 0}$ for all $g \in M'$. Let $g^*$ be the most valuable good in $M'$ according to $i$. Let $\gamma > 0$ be some real number such that $\max_{g \in M'} \frac{\alpha_{g^*} + \alpha_g}{\beta_{g^*} + \beta_g} \le \gamma$. Then,
\begin{align*}
    \frac{\sum_{g \in M'} \alpha_g v_i(g)}{\sum_{g \in M'} \beta_g v_i(g)} \le \gamma
\end{align*}
\end{lemma}
\begin{proof}
Let $\delta = \beta_{g^*} - \frac{\alpha_{g^*}}{\gamma}$. Note that $\delta \ge 0$ since $\gamma \geq \frac{\alpha_{g^*} + \alpha_{g^*}}{\beta_{g^*} + \beta_{g^*}} = \frac{\alpha_{g^*}}{\beta_{g^*}}$. Moreover for any $g \ne g^*$, 
\begin{align*}
 \delta + \beta_g = \beta_{g^*} + \beta_g - \frac{\alpha_{g^*}}{\gamma} \ge \frac1\gamma (\alpha_{g^*} + \alpha_{g} - \alpha_{g^*}) \ge  \frac{\alpha_g}{\gamma}.   
\end{align*}

Using these two observations, we simplify $\sum_{g \in M'} \beta_g v_i(g)$ as follows:
\begin{align*}
    \sum_{g \in M'} \beta_g v_i(g) &= \left (\delta + \frac{\alpha_{g^*}}{\gamma} \right )v_i(g^*) + \sum_{g \in M' \setminus {g^*}} \beta_g v_i(g) \\
    &\ge \frac{\alpha_{g^*}}{\gamma} v_i(g^*) + \delta v_i(M' \setminus \{g^*\}) + \sum_{g \in M' \setminus {g^*}} \beta_g v_i(g) \\
    &= \frac{\alpha_{g^*}}{\gamma} v_i(g^*) + \sum_{g \in M' \setminus {g^*}} (\delta + \beta_g) v_i(g) \\
    &\ge \frac{\alpha_{g^*}}{\gamma} v_i(g^*) + \sum_{g \in M' \setminus {g^*}} \frac{\alpha_g}{\gamma} v_i(g) = \frac1\gamma \sum_{g \in M'} \alpha_g v_i(g),
\end{align*}
which proves the lemma.
\end{proof}

We are ready to prove the main result of this section.

\KTwo*
\begin{proof}
The ex-post \EFX{}+\PO{} guarantee follows from the arguments of \Cref{prop:ex-post utse}. We only prove the ex-ante \EF{} guarantee.

Fix two agents $i$ and $j$. Let agent $i$ consume $x$ amount of $g_i$, and agent $j$ consume $y$ amount of $g_j$. We show that $\E[v_i(Z_j)] - \E[v_i(Z_i)] \le 0.1 \E[v_i(Z_i)]$. Re-arranging the terms of this inequality proves the ex-ante \EF{} guarantee.

Our approach is to upper bound $\E[v_i(Z_j)]$ and lower bound $\E[v_i(Z_i)]$ similar to \Cref{lem:ij-fractional}; the only difference is that we replace the weak bound on the probability of receiving a good in the tail with the exact probability from \Cref{lem:depround-probabilities}.

Following the same notation as \Cref{lem:ij-fractional}, if $x > y$, let $g^*$ be the highest valued good (according to $i$) that $j$ consumed in the interval $[1-x, 1-y]$ in the generation of the fractional allocation $X$. Then, 
\begin{align*}
    \E[v_i(Z_j)] &\le v_i(X_j(1-\max\{x, y\})) + \max\{x-y, 0\}v_i(g^*) \\ &\qquad + \left (y - \frac12\Pr[E_{jg_j}] \right)v_i(g_j) + \frac12 \sum_{g \in L \cup U \setminus \{g_j\}}\left (y - \Pr[E_{jg}] \right) v_i(g) \\
    &= v_i(X_j(1-\max\{x, y\})) + \max\{x-y, 0\}v_i(g^*) + \frac12 yv_i(g_j) + \frac12 \sum_{g \in L \cup U}\left (y - \Pr[E_{jg}] \right) v_i(g).
\end{align*}

Similarly, for agent $i$, we get
\begin{align*}
    \E[v_i(Z_i)] &\ge v_i(X_j(1-\max\{x, y\})) + \max\{y-x, 0\}v_i(L \cup U) \\ &\qquad + \left (x - \frac12\Pr[E_{ig_i}] \right)v_i(g_i) + \frac12 \sum_{g \in L \cup U \setminus \{g_i\}}\left (x - \Pr[E_{ig}] \right) v_i(g) \\
    &\ge v_i(X_i(1-\max\{x, y\})) + \left (\max\{x, y\} - \frac12\Pr[E_{ig_i}] \right)v_i(g_i) \\ & \qquad+ \frac12 \sum_{g \in L \cup U \setminus \{g_i\}}\left (\max\{x,y\} - \Pr[E_{ig}] \right) v_i(g) \\
    &= v_i(X_i(1-\max\{x, y\})) + \frac12 \max\{x, y\}v_i(g_i) + \frac12 \sum_{g \in L \cup U}\left (\max\{x,y\} - \Pr[E_{ig}] \right) v_i(g).
\end{align*}

Taking the difference and using $v_i(X_i(1-\max\{x, y\})) \ge  v_i(X_j(1-\max\{x, y\}))$,
\begin{align}
    \E[v_i(Z_j)] - \E[v_i(Z_i)] &\le \frac12yv_i(g_j) - \frac12 \max\{x, y\}v_i(g_i) +\max\{x - y, 0\} v_i(g^*) \notag \\&\qquad + \frac12 \sum_{g \in L\cup U} (y - \Pr[E_{jg}] - \max\{x, y\} + \Pr[E_{ig}])v_i(g). \label{eq:diff-upper-bound-1}
\end{align}

We obtain a lower bound on $\E[v_i(Z_i)]$ from \Cref{lem:depround-lower-bound} (using $\hat{g} = g^*$), which we re-state here for readability.
\begin{align}
    \E[v_i(Z_i)] \ge \left (1- \frac12x(1-x) \right )v_i(g_i) + (1-x)v_i(L \cup U \cup \{g^*\} \setminus \{g_i\}).\label{eq:diff-lower-bound-1}
\end{align}

So far, we have followed the same procedure as \Cref{thm:utse-upper-bounnd}. However, \eqref{eq:diff-upper-bound-1} cannot be simplified further without hurting the ex-ante \EF{} guarantee. To upper bound the ratio $\frac{\E[v_i(Z_j)] - \E[v_i(Z_i)]}{\E[v_i(Z_i)]}$, we use \Cref{lem:upperbound-trick} to upper bound the ratio of the expressions \eqref{eq:diff-upper-bound-1} and \eqref{eq:diff-lower-bound-1}. In the parlence of \Cref{lem:upperbound-trick}, $M' = L \cup U \cup \{g^*\}$, and the highest valued good in $L \cup U \cup \{g^*\}$ according to $i$ is $g_i$. We also refer to the coefficient of good $g$ in \eqref{eq:diff-upper-bound-1} as $\alpha_g$ and the coefficient of good $g$ in \eqref{eq:diff-lower-bound-1} as $\beta_g$. We prove the following lemma. 

\begin{lemma}
For any $g \in L \cup U \setminus g_i$, $\alpha_g + \alpha_{g_i} \le\max_{z \ge x} \frac12 z(1-z)$ if $g \ne g_i$, and $\alpha_{g_i} \le \frac{1}{2} x(1-2x)$.
\label{lem:coeff-bounds}
\end{lemma}

Before we prove the lemma, we show how these upper-bounds immediately give us the result. Using \Cref{lem:upperbound-trick} with the $\beta$ values from \eqref{eq:diff-lower-bound-1}, we get
\begin{align*}
    \frac{\E[v_i(Z_j)] - \E[v_i(Z_i)]}{\E[v_i(Z_i)]} &\le \max_{g \in L \cup U \cup \{g^*\}} \left \{\frac{\alpha_{g_i} + \alpha_g}{\beta_{g_i} + \beta_g} \right \} \\
    &\le \max_{x \in [0, 1]} \max_{z \ge x}\max \left \{\frac{\frac12z(1-z)}{2-x-\frac12x(1-x)}, \frac{\frac12x(1-2x)}{1-\frac12x(1-x)} \right \}\le 0.094. 
\end{align*}

The last inequality can be verified using basic calculus (or a graphing calculator). 

\begin{proof}[Proof of \Cref{lem:coeff-bounds}]
Since we do not know exactly which item $g_j$ is, we use an indicator variable which takes value $1$ if $g_j$ satisfies certain conditions. For example, $\mathbb{I}_{\{g = g_j\}}$ takes value $1$ when $g = g_j$. We divide the proof into a few cases. 

\textbf{Case 1: $g = g^*$}
\begin{align*}
    \alpha_{g_i} + \alpha_{g^*} &\le \frac12 (2\max\{x-y, 0\} + y - \Pr[E_{jg_i}] - 2\max\{x, y\} + \Pr[E_{ig_i}] + \mathbb{I}_{\{g_i = g_j\}}y) \\
    &\le \frac12\Pr[E_{ig_i}] \le \frac12 x(1-x). 
\end{align*}

\textbf{Case 2: $g \in L \cup U \setminus \{g_i\}$}
\begin{align*}
    \alpha_{g_i} + \alpha_{g} &= \frac12 (2y - \Pr[E_{jg_i}] -\Pr[E_{jg}] - 3\max\{x, y\} + \Pr[E_{ig_i}] + \Pr[E_{ig}] + y\mathbb{I}_{\{g_j \in \{g, g_i\}\}}) \\
    &\le \frac12 [(y - \Pr[E_{jg_i}] - \Pr[E_{jg}]) - 2\max\{x, y\} + \Pr[E_{ig_i}] + \Pr[E_{ig}] + y \mathbb{I}_{\{g_j \in \{g, g_i\}\}}].
\end{align*}
The first term $y - \Pr[E_{jg_i}] - \Pr[E_{jg}]$ is the probability that $j$ receives the super-good together with some agent whose last consumed good is neither $g_i$ nor $g$. Using \Cref{lem:depround-basic}, this probability is at most $y(2-q(g_i) - q(g))$ where $q(g')$ denotes the amount of $g'$ consumed during the eating process.
Similarly, we can show that $\Pr[E_{ig_i}] \le x(q(g_i) - x)$ and $\Pr[E_{ig}] \le xq(g)$. Putting this together, we get:
\begin{align*}
     \alpha_{g_i} + \alpha_{g} &\le \frac12 [y(2-q(g_i) - q(g)) + x(q(g_i) + q(g) - x) - 2\max\{x, y\} + y\mathbb{I}_{\{g_j \in \{g, g_i\}\}}]   \\
     &= \frac12 [2y - x^2 + (x-y)(q(g_i) + q(g))- 2\max\{x, y\} + y\mathbb{I}_{\{g_j \in \{g, g_i\}\}}].
\end{align*}
Note that $q(g_i) + q(g)$ is at least $x + y\mathbb{I}_{\{g_j \in \{g, g_i\}\}}$; it is also trivially at most $2$. Depending on the sign of $x-y$, we can choose an appropriate bound. If $x > y$, 
\begin{align*}
    \alpha_{g_i} + \alpha_{g} &\le \frac12 [2y - x^2 + (x-y)(2) - 2\max\{x, y\} + y\mathbb{I}_{\{g_j \in \{g, g_i\}\}}]  \\
    &= \frac12(2x - x^2 - 2\max\{x, y\} + y\mathbb{I}_{\{g_j \in \{g, g_i\}\}}) \le \frac12 x(1-x).
\end{align*}

If $x \le y$,
\begin{align*}
    \alpha_{g_i} + \alpha_{g_j} &\le \frac12 [2y - x^2 + (x-y)(x+y\mathbb{I}_{\{g_j \in \{g, g_i\}\}}) - 2\max\{x, y\} + y\mathbb{I}_{\{g_j \in \{g, g_i\}\}}] \\
    &= \frac12 [2y - xy - 2\max\{x, y\} + y (1 + x- y) \mathbb{I}_{\{g_j \in \{g, g_i\}\}}] \\
    & \le \frac12 [2y - xy - 2\max\{x, y\} + y (1 + x- y)] = \frac12(y -y^2) \le \max_{z \ge x} \frac12 z(1-z).
\end{align*}

The second inequality follows from the fact that $y(1+x - y) \ge 0$.

\textbf{Case 3: $g = g_i$}
\begin{align*}
    \alpha_{g_i} \le \frac12 (y - \Pr[E_{jg_i}] - 2\max\{x, y\} + \Pr[E_{ig_i}] + y\mathbb{I}_{\{g_j = g_i\}}).
\end{align*}

This is non-positive if $g_j \ne g_i$ since $\Pr[E_{ig_i}] \le x$. Therefore, assume $g_j = g_i$. 
The value $(\Pr[E_{ig_i}] - \Pr[E_{jg_i}])$ is at most the probability that $i$ is allocated the super-item together with an agent in $A_{g_i}$ who is not $j$. This probability is at most $x(1- x - y)$ using \Cref{lem:depround-basic}.

Therefore, 
\begin{align*}
    \alpha_{g_i} &\le \frac12 [x(1-x-y) + 2y - 2\max\{x, y\}] \\
    &\le \frac12 [x(1-x-y) + 2yx - 2x\max\{x, y\}] \\
    &\le \frac12 [x(1 - x - \max\{x, y\})] \le \frac12x(1-2x).
\end{align*}

This completes the proof of \Cref{lem:coeff-bounds}.
\end{proof}

\Cref{thm:k-two} now follows from \Cref{lem:coeff-bounds}.
\end{proof}

Combining \Cref{thm:utse-upper-bounnd,thm:low-k,thm:k-two}, we get the following theorem. 

\dependentrounding*

\section{Asymptotic Result for Lexicographic goods}\label{apdx:1-o(1)}

To show the contrast between Algorithms \ref{alg:1-o(1)} and \ref{alg:dependent-rounding}, we show in the following example, an instance where Algorithm \ref{alg:1-o(1)} cannot perform better than ex-ante $\frac{4}{5}$-\EF{}, while our Dependent Rounding Algorithm still preserves the guarantee of ex-ante $\frac{9}{10}$-\EF{}. 

\begin{example}[Algorithm \ref{alg:1-o(1)} fails ex-ante $(\frac{4}{5}+\eps)$-\EF{}] \label{app:1-o(1) example}
\begin{proof}

Consider an instance with three goods and three agents. For $\eps > 0$ sufficiently small, the valuations are given as follows:
\begin{table}[h]
    \centering
    \begin{tabular}{c|ccc}
        & $g_1$ & $g_2$ & $g_3$\\  
        \cmidrule{1-4}
        $1$ & $1+2\epsilon$ & $1$ & $\epsilon$ \\  
        $2$ & $\epsilon$ & $1+2\epsilon$ & $1$ \\
        $3$ & $1+2\epsilon$ & $\epsilon$ & $1$ \\
        
    \end{tabular}
    \label{tab: 0.8 ex-ante EF example}
\end{table}

Upon eating for $1-\frac{1}{n}$ units of time, we get the following fractional matrix $X$. 
\begin{table}[h]
    \centering
    \begin{tabular}{c|ccc}
        & $g_1$ & $g_2$ & $g_3$\\  
        \cmidrule{1-4}
        $1$ & $1/2$ & $1/6$ & $0$ \\  
        $2$ & $0$ & $2/3$ & $0$ \\
        $3$ & $1/2$ & $0$ & $1/6$ \\
    \end{tabular}
\end{table}

$X$ can be decomposed into integral allocations $A^1,\,A^2,\,A^3,\,A^4$ as shown below.  

$$\NiceMatrixOptions{code-for-first-row=\scriptstyle,code-for-first-col=\scriptstyle}
\begin{bNiceMatrix}
1/2 & 1/6 & 0 \\  
0 & 2/3 & 0 \\
1/2 & 0 & 1/6
\end{bNiceMatrix}
=
\frac{1}{3} \cdot
\begin{bNiceMatrix}
0 & 0 & 0 \\  
0 & 1 & 0 \\
1 & 0 & 0
\end{bNiceMatrix} 
+
\frac{1}{6} \cdot
\begin{bNiceMatrix}
 1 & 0 & 0 \\
 0 & 0 & 0 \\ 
 0 & 0 & 1 
\end{bNiceMatrix} 
+
\frac{1}{6} \cdot
\begin{bNiceMatrix}
 0 & 1 & 0 \\
 0 & 0 & 0 \\ 
 1 & 0 & 0
\end{bNiceMatrix} 
+
\frac{1}{3} \cdot
\begin{bNiceMatrix}
 1 & 0 & 0 \\
 0 & 1 & 0 \\ 
 0 & 0 & 0 
\end{bNiceMatrix} 
$$

In each of these integral allocations, exactly one agent receives a singleton tail. 
\begin{align*}
     A^1: \text{agent 1 receives } g_3 \\
     A^2: \text{agent 2 receives } g_2 \\
     A^3: \text{agent 2 receives } g_3 \\
     A^4: \text{agent 3 receives } g_3 
\end{align*}
Let us analyze agent 1's envy towards agent 2. 
\begin{align*}
    \E[v_1(\Z_1)] = \frac{1}{3}v_1(g_3)+\frac{1}{6}v_1(g_1)+\frac{1}{6}v_1(g_2)+\frac{1}{3}v_1(g_1) = \frac{8\epsilon+4}{6} \\ 
    \E[v_1(\Z_2)] = \frac{1}{3}v_1(g_2)+\frac{1}{6}v_1(g_2)+\frac{1}{6}v_1(g_3)]+\frac{1}{3}v_1(g_2) = \frac{5+\epsilon}{6}
\end{align*}

Taking the ratio of agent 1's valuation for its own bundle with agent 2's bundle under $\Z$:
$$\frac{\E[v_1(\Z_1)]}{\E[v_1(\Z_2)]} = \frac{4+8\epsilon}{5+\epsilon} \leq \frac{4}{5}+\epsilon$$ 
\end{proof}
\end{example}

\section{Decomposing Probabilistic Serial Outcome Over Picking Sequences}
\label{subsec:Proof_sdEF_EF1+PO_lexicographic}

In this section, we will show that the outcome of the Probabilistic Serial algorithm can be written as a probability distribution over deterministic allocations each of which is induced by a picking sequence. A picking sequence is a sequence $(a_1, a_2, \dots,a_m)$ of length $m$, where $a_i \in N$ for every $i \in [m]$. A picking sequence naturally induces an allocation as follows: agents take turns following the sequence, with each picking its favorite remaining item. For example, given an instance with four goods, the sequence $\sigma = (1, 2, 3, 2)$ corresponds to an allocation where agents $1$, $2$, and $3$ pick their favorite remaining good in that order, and finally agent~$2$ picks the leftover good.

\begin{restatable}{proposition}{PSandPickingSequence}
Let $X$ denote the fractional allocation computed by the Probabilistic Serial algorithm. Then, there exists a randomized allocation $\X \coloneqq \{(p_1,A^1),$ $\dots, (p_s,A^s)\}$ such that $X = \sum_{k \in [s]} p_k A^k$ and each deterministic allocation $A^k$ is the outcome of some picking sequence.
\label{prop:PS_Picking_Sequence}
\end{restatable}

Before presenting the proof of \Cref{prop:PS_Picking_Sequence}, let us discuss a decomposition technique due to Aziz et al.~\cite{AFS+24best} which will be useful in proving some of the intermediate results required by the proof.

Let us assume for simplicity that the number of goods is an integral multiple of the number of agents, i.e., $m = rn$ for some $r \in \mathbb{N}$. (This assumption can be shown to hold without loss of generality by first adding dummy zero-valued goods to ensure this condition and later removing them without affecting any of the guarantees discussed below.) Thus, the Probabilistic Serial algorithm finishes after $r$ units of time (or $r$ \emph{rounds}).\footnote{The ``rounds'' terminology is used merely for ease of exposition, and should not be considered to mean that we restart the algorithm from scratch after each round.}

For each $k \in [r]$, we define $i_k$ to be the `representative' of agent $i$ during the $k^\text{th}$ round (i.e., during the time interval $[k-1,k)$). The preferences of the representative $i_k$ are the same as that of agent $i$, but is only allowed to eat on behalf of agent $i$ during the $k^\text{th}$ round. Thus, the eating trajectory of agent $i$ can be simulated by concatenating that of each representative agent $i_1,i_2,\dots,i_r$ in the corresponding time window.

Let $X$ denote the fractional allocation returned by the Probabilistic Serial algorithm. That is, $X$ is an $n \times m$ matrix where $X_{i,j}$ denotes the fraction of good $j$ consumed by $i$. 

We shall now represent the outcome of Probabilistic Serial algorithm using an $m \times m$ matrix $Y$ instead. The columns of $Y$ correspond to the goods, and its rows correspond to the representative agents, where the rows $1,\dots,n$ correspond to the representatives $1_1,\dots,n_1$ for the first unit of time, the rows $n+1,\dots,2n$ correspond to the representatives $1_2,\dots,n_2$ for the second unit of time, and so on. The entry $Y_{i,j}$ denotes the fraction of good $j$ eaten by the agent that $i$ represents during the time interval $[k-1,k)$ where $k = \lceil i/n \rceil $.  Observe that $Y$ is a doubly stochastic matrix: The entries in each column sum up to $1$ since there is a unit amount of each good available, and the entries in each row sum up to $1$ since each representative agent eats at the speed of one item per unit time for the duration of one unit of time.

Consider any Birkhoff-von Neumann decomposition of $Y$ into permutation matrices given by
$$Y = p_1 \cdot A^1 + p_2 \cdot A^2 + \dots + p_s \cdot A^s.$$

Notice that for any $\ell \in [s]$, the permutation matrix $A^\ell$ induces a deterministic complete allocation wherein each representative agent receives exactly one good. Thus, $Y$ can be equivalently represented as a randomized allocation $\Y \coloneqq \{(p_1,A^1),$ $\dots,(p_s,A^s)\}$.

Let us now discuss the proof of \Cref{prop:PS_Picking_Sequence}.

\begin{proof} (of \Cref{prop:PS_Picking_Sequence})
Consider any fixed allocation $A^\ell$ in the support of the randomized allocation $\Y$. With every good $g \in M$, let us associate a time instant $t_g$ such that if $g \in A^\ell_{i_k}$, then $t_g$ is the time at which the representative $i_k$ started to consume $g$ for the first time during the execution of the Probabilistic Serial algorithm. For notational convenience, let us reindex the goods so that for any $i,j \in [m]$ with $i \leq j$, we have $t_{g_i} \leq t_{g_j}$ (i.e., the owner of a higher-index good in $A^\ell$ cannot start consuming it until the owner of a lower-index good in $A^\ell$ has started to do so). Let us also define an \emph{ownership} function~$o: M \rightarrow R$, where $R$ denotes the set of representative agents, such that for any good $g \in M$, $o(g) \coloneqq i_k$ if the representative agent $i_k$ is the owner of good $g$ under $A^\ell$, i.e., $g \in A^\ell_{i_k}$.

We will argue that the picking sequence $\sigma \coloneqq \langle o(g_1),o(g_2),\dots,o(g_m) \rangle$ induces the allocation $A^\ell$. We will prove this by induction.

To see why the base case is true, observe that at the time $t_{g_1}$, all the other goods in $M \setminus \{g_1\}$ are at least partially available. This is because under the Birkhoff-von Neumann decomposition of $\Y$, a representative agent is allocated a good in $A^\ell$ only if it eats a nonzero portion of this good during the execution of Probabilistic Serial algorithm on the original instance. Since $t_{g_1}$ is the earliest time instant at which any representative agent starts eating the good its receives under $A^\ell$, it must be that a nonzero amount of every other good is available at time $t_{g_1}$.

Recall that under the Probabilistic Serial algorithm, at every time instant, each agent consumes its favorite available good. Since the representative agent $o(g_1)$ starts to consume the good $g_1$ when all other goods are at least partially available, it must weakly prefer the good $g_1$ over any other good in $M \setminus \{g_1\}$. This implies that under the sequence $\sigma$ where $o(g_1)$ goes first, it will pick the good $g_1$.

Now suppose that for some $j < m$ and for every $i \leq j$, the agent $o(g_i)$ picks the good $g_i$ on its turn under $\sigma$. We will now argue that the representative agent $o(g_j)$ will pick the good $g_j$ on its turn.

Indeed, when it is agent $o(g_j)$'s turn in $\sigma$, the set of available goods is $M \setminus \{g_1,\dots,g_{j-1}\}$, or, equivalently, the set $\{g_j,g_{j+1},\dots,g_m\}$ (this follows from the induction hypothesis). All of these goods are at least partially available at the time $t_{g_j}$. Since the agent $o(g_j)$ starts to consume $g_j$ at time $t_{g_j}$, it must be that it weakly prefers $g_j$ over any other good $g_{j'}$ such that $j < j'$. Thus, it must be that on its turn under $\sigma$, agent $o(g_j)$ picks the good $g_j$.

We therefore obtain that the allocation $A^\ell$ is induced by the picking sequence $\sigma$, as desired.
\end{proof}

\begin{restatable}{remark}{eating_1_unit_time} \label{rem:eating_1_unit_time}
Note that the above arguments continue to hold even when the Probabilistic Serial algorithm is run for exactly one unit of time, i.e., we consider each agent's representative from only the $1^{st}$ round along with the goods it has consumed and apply the same induction argument. 
\end{restatable}

Thus, the fractional allocation computed by running Probabilistic Serial for exactly one time unit can also be decomposed into a probability distribution over integral partial allocations each of which is induced by a picking sequence. In all of these integral allocations, each agent receives exactly one good.

\end{document}